\documentclass{article}

\usepackage[inner=3cm,outer=3cm,top=3cm,bottom=3cm]{geometry} 
\usepackage[utf8]{inputenc}
\usepackage{amsmath} 
\usepackage{amsthm} 
\usepackage{amssymb} 
\usepackage{mathtools} 
\usepackage{tikz} 
\usepackage{graphicx} 
\usepackage{tikz-cd} 
    \usetikzlibrary{cd}
\usepackage{mathrsfs} 
\usepackage{array} 
\usepackage{setspace} 
\usepackage{xcolor} 
\usepackage{harvard} 
\usepackage{enumitem} 
\usepackage{tcolorbox} 
\usepackage{float}
\usepackage{pdfpages} 
\usepackage{subcaption} 
\usepackage{hyperref} 
\hypersetup{colorlinks=true,linkcolor=blue,citecolor=blue}
\usepackage{dsfont} 
\usepackage{tikz}

\hypersetup{
    colorlinks=true,
    linkcolor={blue!50!black},
    citecolor={blue!50!black},
    urlcolor={blue!50!black} 
}

\theoremstyle{definition}
\newtheorem{definition}{Definition}
\theoremstyle{plain}
\newtheorem{lemma}{Lemma}
\theoremstyle{plain}
\newtheorem{theorem}{Theorem}
\theoremstyle{plain}

\theoremstyle{definition}
\newtheorem{example}{Example}
\theoremstyle{definition}

\theoremstyle{plain}
\newtheorem{corollary}{Corollary}
\theoremstyle{plain}

\theoremstyle{plain}

\DeclareMathOperator*{\ptop}{top}
\DeclareMathOperator*{\pbottom}{bottom}
\DeclareMathOperator*{\id}{id}

\definecolor{ForestGreen}{rgb}{.13,.54,.13}

\title{Royal Processions: Incentives, Efficiency and Fairness in Two-sided
Matching}

\author{Sophie Bade\thanks{Royal Holloway, University of London and Max Planck Institut for Research on Collective Goods, Email: sophie.bade@rhul.ac.uk} \and Joseph Root\thanks{Department of Economics, University of Chicago, Email: jroot@uchicago.edu}}

\date{January 2023}

\begin{document}

\maketitle

\begin{abstract}
We study the set of incentive compatible and efficient two-sided matching mechanisms. We classify all such mechanisms under an additional assumption -- ``gender-neutrality" -- which guarantees that the two sides be treated symmetrically. All group strategy-proof, efficient, and gender-neutral mechanisms are recursive and the outcome is decided in a sequence of rounds. In each round two agents are selected, one from each side. These agents are either ``matched-by-default" or ``unmatched-by-default." In the former case either of the selected agents can unilaterally force the other to match with them while in the latter case, they may only match together if both agree. In either case, if this pair of agents is not matched together, each gets their top choices among the set of remaining agents. As an important step in the characterization, we first show that in one-sided matching all group strategy-proof and efficient mechanisms are sequential dictatorships. An immediate corollary is that there are no individually rational, group strategy-proof and efficient one-sided matching mechanisms.  
\end{abstract}

\section{Introduction}

\citeasnoun{gale1962college} 
initiated the study of stability in two-sided matching problems using the example of a marriage market where each of $n$ men is matched to one of $n$ women. A matching is considered stable if no unmarried pair prefers each other to their partners from the matching. 
\citeasnoun{gale1962college} proved that stable matchings always exist. However, \citeasnoun{roth1982economics} showed that stability and incentive compatibility clash: any stable mechanism is manipulable.

Faced with this dilemma, the vast literature on two-sided matching typically sacrifices incentive compatibility in favor of stability.\footnote{See for instance \citeasnoun{roth1992two}, \citeasnoun{roth2008deferred} and \citeasnoun{abdulkadiroglu2013matching}}  In this paper, we do the opposite: we ignore stability and instead study incentive compatible mechanisms. Stability is an especially important criterion for settings where agents can easily rearrange matchings among themselves \cite{abdulkadiroglu2013matching}. However there are many settings where coalitional deviations may be hard to organize. For example, people who sign up for a dating app do so since they find it difficult to find dates in real life. 
In other settings the central planner may have the authority to enforce the outcome of a mechanism. Take, for example, the British ``homes for Ukraine'' scheme that matches refugees with host families. The UK home office could condition residency permits on the matches found by an algorithm. So even if some families and refugees would like to deviate from the matching, they might find it impossible to do so. Finding mechanisms with strong incentive properties so that they are simple to use would, in both these contexts, seem more important.

By abandoning stability, we open the door to the study of incentive compatible and efficient mechanisms. In particular, we look for mechanisms with robust incentive properties -- those that are group strategy-proof. A mechanism is group strategy-proof if there is no profile of preferences, group of agents and deviation for that group, such that all group members are weakly better-off after the deviation (and some strictly). A mechanism is efficient if it never chooses a matching where an alternative could make all agents better-off.
Group strategy-proofness and efficiency are easy to attain in two-sided matching problems. Simply declare one side of the market to be ``objects" and use one of the well-known group strategy-proof and efficient mechanisms for allocation problems.\footnote{The set of all group strategy-proof and efficient mechanism for the house allocation problem was characterized by \citeasnoun{pycia2017incentive} and \citeasnoun{bade2020random}. The characterization extends Gale's top trading cycles \cite{shapley1974cores} in three ways: agents may ``own" multiple houses \cite{papai2000strategyproof}, the agents may ``broker" houses \cite{pycia2017incentive} and when there are exactly three agents and houses there may be ``braids" \cite{bade2020random}. }
Indeed sometimes ``objectification" of one side of the market may be appropriate. \citeasnoun{abdulkadirouglu2003school}, for example, make a convincing case to treat schools as objects in school choice problems. 
However in settings where the two sides are symmetric a priori, one may be hesitant to assign all agency to one side of the market. In fact, we may want to require that the two sides of the market are treated equally. There is, for example, no reason to treat the men and women on a dating app any differently. Likewise one may wish to give an equal amount of say to the refugees and the host families in the ``Homes for Ukraine" program.

Our axiom of ``gender-neutrality" captures the idea that the two sides should be treated equally. Roughly, we require that if the men and women's preferences are swapped, the outcome of the mechanism is swapped as well.
We not only impose such symmetry between the sides in the overall mechanism but on every ``continuation submechanism" as well. The notion of ``continuation submechanism" captures the idea that a mechanism can be recursive. If a mechanism makes initial matches using only the preferences of the initially matched agents, ignoring all other agents preferences, then the remaining agents face such a continuation submechanism and we ask that they be treated symmetrically as well.

We characterize the set of all group strategy-proof, efficient and gender-neutral mechanisms as the set of ``royalty mechanisms." These work as follows: each round one agent is chosen from either side of the market. These ``royals"  are either ``matched-by-default" or ``unmatched-by-default." In the former case the royals are matched together if at least one of them top-ranks the other; in the latter case, the royals are only matched if both top-rank one another. If the royals are not matched with each other, they get their top choices among the set of remaining agents. As long as at least 6 agents remain unmatched, a procession of choices by such royal couples continues. Once only four agents remain, there are only two possible matches for these agents, however a plethora of (sub)mechanisms for these remaining agents satisfy our desiderata. 

Our proof begins with the observation that any gender-neutral two-sided mechanism nests a one-sided matching mechanism, where any agent is free to match with any other agent. To get an intuition for this result, index the men and women such that $m_i$ is the man who mirrors women $w_i$ under the desired symmetry. 
If we restrict attention to preference profiles where $m_i$ and $w_i$ announce symmetric preferences, by gender-neutrality they must get symmetric outcomes. So for each $i$, a symmetric mechanism either matches $w_i$ and $m_i$ with each other or with a different symmetric pair $m_j$ and $w_j$. We can therefore treat the pairs $(m_i,w_i)$ as the agents in a one-sided matching problem. Any pair either matches with themselves, which is analogous to the pair remaining single, or they swap partners with another symmetric pair, which can be viewed as a match between these two pairs. The embedded mechanism inherits the group strategy-proofness and efficiency of the original mechanism. This is helpful since the group strategy-proof and efficient mechanisms admit a neat characterization for the one-sided matching problem. We show that all such mechanisms are ``sequential dictatorships" where sequences of agents or ``dictators" choose their most preferred remaining agent as their partner.
This characterization complements that of \citeasnoun{root2020incentives} which finds a similar result when agents are not allowed to remain single. Despite the similarities in the result, the proofs are very different. Our proof stems from our analysis of the three-agent case which is ruled out in \citeasnoun{root2020incentives} since all agents are required to be matched. Having established that the embedded one-sided mechanism is a sequential dictatorship, we see that any gender-neutral, efficient and group strategy-proof mechanism agrees with a royalty mechanism on all symmetric preference profiles. Our next task is to show that the same holds for asymmetric profiles. The proof is involved, but boils down to finding a path from any asymmetric preference profile to a symmetric profile, over which we can show the mechanism agrees with a royalty mechanism.

Gender-neutrality can be seen as a fairness requirement. A common fairness criterion in social choice, full symmetry, requires that all agents are treated identically, in the sense that the outcome of a mechanism is invariant under any bijective relabeling of all agents.
In many environments full symmetry conflicts with other desiderata.  \citeasnoun{bartholdi2021equitable} therefore weaken full symmetry to require that the outcome of a mechanism must only remain unchanged for a particular set of bijective relabelings of agents.  From this point of view, our notion of (weak) gender-neutrality only requires that the mechanism be symmetric with respect to \emph{one} such relabeling. In fact, in two-sided matching, no efficient and group strategy-proof mechanism is invariant under more than two relabelings. 
 
Faced with these limits on achieving symmetry we consider randomization. We first show that randomization is no panacea. In our environment, just as in house allocation where random mechanisms are well-studied, randomization can introduce inefficiency. A typical approach is to ``symmetrize" a mechanism by choosing agents' roles in the mechanism uniformly at random. For example, ``serial dictatorship", where agents choose their match sequentially in a fixed order, can be symmetrized to yield ``random serial dictatorship" by selecting the picking order uniformly at random. Despite the efficiency of serial dictatorship, random serial dictatorship can be inefficient. Nevertheless, the symmetrized mechanism retains the incentive properties of the deterministic mechanism from which it is derived. In house allocation, no mechanism can simultaneously achieve efficiency, incentive compatibility and symmetry \cite{bogomolnaia2001new}. On the other hand, random serial dictatorship admits a certain form of efficiency: it selects a lottery over deterministic outcomes each of which is efficient. We say that random serial dictatorship is therefore ex-post efficient. A major open question in house allocation is whether any other mechanism satisfies ex-post efficiency, incentive compatibility and symmetry. We show that two-sided matching admits multiple mechanisms of this type. 

The rest of the paper proceeds as follows. We first establish the preliminary definitions then describe our results. Finally we discuss our axioms. We conclude with a discussion of randomized mechanisms.

\section{Preliminaries}\label{sec:preliminaries}

Let $N$ be a finite set of agents. We will be interested in both ``one-sided matching", where any agent can match with any other and ``two-sided matching", where matching is bipartite, so that agents can only match with partners from the other side of the market. In two-sided matching, we will assume that $N$ is the disjoint union of two sets, $M=\{m_{1},\dots, m_{n}\}$ and $W=\{w_{1},\dots, w_{n}\}$. In keeping with \citeasnoun{gale1962college}, we refer to the agents in $M$ as ``men" and the agents in $W$ as ``women."
A \textbf{submatching} is a (possibly empty) list of mutually exclusive pairs and singletons. The pairs are unordered so that $(m,w)$ and $(w,m)$ both refer to the same pair. In the case of two-sided matching singletons are not allowed and each pair has to be made up of one man and one woman. Let $\Sigma^1_0$ and $\Sigma^2_0$ denote the set of one- and two-sided submatchings respectively. For any submatching $\nu$ let $N(\nu)$ denote the set of agents matched in the submatching. A \textbf{matching} is a submatching that lists every agent. A \textbf{proper submatching} is a submatching which is not a matching. We alternatively represent matchings as bijections $\mu:N\to N$ of order 2 (i.e. $\mu\circ \mu=\id$) so that agent $i$ is matched with $j$ if and only if $j$ is matched with $i$.  The requirement that pairs must be made up of one man and one woman each then translates to  $\mu(M)= W$. We denote the sets of one- and two-sided matchings by $\Sigma^1$ and $\Sigma^2$ respectively.

Agents are assumed to have strict preferences (total orders) over their possible partners. In one-sided matching, each agent $i$ has a strict preference $\succsim_{i}$ over $N$, where $i$ stands for the option to stay unmatched. We write $x\succ_i y$ to mean that $i$ strictly prefers being matched with $x$ to being matched with $y$. We write $x\succsim_i y$ to indicate that either $x\succ_i y$ or $x=y$. In two-sided matching, an agents' preferences range over the other side of the market. That is, each $m\in M$ has a strict preference $\succsim_{m}$ over $W$ and each $w\in W$ has a strict preference $\succsim_{w}$ over $M$. A preference profile is a list all agents' preferences $(\succsim_{i})_{i\in N}$. Let $\Omega^{1}$ and $\Omega^2$ denote the set of preference profiles for the one-sided and two-sided matching problems respectively.

For any preference $\succsim_{i}$ in either context, define $\ptop(\succsim_i)$ as agent $i$'s most preferred partner and $\pbottom(\succsim_i)$ as agent $i$'s least preferred partner. For any agent $i$, we use the notation $\succsim_{i}:j_{1}, j_{2}, \cdots $ to mean an arbitrary preference which top-ranks $j_{1}$, second-ranks $j_{2}$ and so on. 
Given a subset of agents $S$, we write $\succsim_{S}$ as a preference profile for just the agents in $S$. Given a preference profile $\succsim$ and some $\succsim_{S}'$, we write $(\succsim_{S}',\succsim_{-S})$ for the preference profile in which each agent $i$ from $S$ announces $\succsim'_{i}$ and each agent $j$ from $N\setminus S$ announces $\succsim_j$.

A \textbf{(matching) mechanism} is a function $f:\Omega^{k} \rightarrow \Sigma^{k}$ for $k=1,2$ that maps each profile of preferences to a matching.
Such a mechanism $f$ is \textbf{group strategy-proof} if there is no preference profile $\succsim$, group of agents $S\subset N$ and $\succsim'_{S}$ such that $f(\succsim_{S}',\succsim_{-S})(i)\succsim_{i} f(\succsim)(i)$ for all $i\in S$, and for some $j\in S$, $f(\succsim_{S}',\succsim_{-S})(j)\succ_{j} f(\succsim)(j)$. That is, if agents' true preferences were captured by the profile $\succsim$ there would be no group of agents who could jointly misreport, making all agents in the group weakly better-off, with at least one strictly better-off.\footnote{It turns out that in this context a mechanism is group strategy-proof if and only if no group of agents of size one or two can find a jointly profitable misreport. See \citeasnoun{alva2017manipulation} and \citeasnoun{root2020incentives} for further discussion.}
Given a preference profile $\succsim$, a matching $\nu$ \textbf{Pareto dominates} a matching $\eta$ if for all agents $i$, $\nu(i)\succsim_{i} \eta(i)$ and for at least one agent $j$, $\nu(j)\succ_{j} \eta(j)$. A matching $\eta$ is  \textbf{efficient} at the preference profile $\succsim$ if there is no $\nu$ which Pareto dominates it.  A mechanism $f$ is \textbf{efficient} if for every preference profile $\succsim$, $f(\succsim)$ is efficient.

We now turn to our fairness requirement. Informally, a mechanism is weakly gender-neutral if there is a way to reflect preferences across the sides so that, after reflection, the mechanism chooses the same outcome, reflected. To formally define weak gender-neutrality 
fix a $\sigma \in \Sigma^{2}$. While $\sigma$ is a matching, we use it here to denote the reflection across which the mechanism exhibits the desired symmetry. Thus $\sigma(w_j)=m_i$ is interpreted to mean that agent $w_j$ is symmetric to agent $m_i$. We use $\sigma$ both to transform preferences and matchings as follows. Given a matching $\mu\in \Sigma^2$, the \textbf{reflection}, $\sigma\ast\mu$, of $\mu$ under $\sigma$ is the matching such that $(m,w)\in \sigma \ast \mu \iff (\sigma^{-1}(m),\sigma^{-1}(w))=(\sigma(m),\sigma(w))\in \mu$.\footnote{The reflection $\sigma \ast \mu$ is the matching given by the function $\sigma \circ \mu \circ \sigma$.}. Equivalently, for any pair, $(m,w)$ matched in $\mu$, $(\sigma(m),\sigma(w))$ are matched in $\sigma \ast \mu$.
For a preference profile $\succsim\in\Omega^2$ we define the reflection $\succsim'=\sigma(\succsim)$ so that $j\succsim'_i j'\Leftrightarrow \sigma(j)\succsim_{\sigma(i)}\sigma(j')$. So if man $m$ prefers woman $w$ to woman $w'$ in the profile $\succsim$, then woman $\sigma(m)$ prefers man $\sigma(w)$ to man $\sigma(w')$ according to the reflected profile $\sigma(\succsim)$. 
Gender-neutrality will ensure that the symmetric agents $i$ and $\sigma(i)$ are  is treated equally: a mechanism $f$ is \textbf{weakly gender-neutral} if $f(\sigma(\succsim))=\sigma \ast f(\succsim)$ holds for all $\succsim\in \Sigma^2$.\footnote{Note that we require $\sigma$ to be of order two. This is without loss of generality in the following sense. One could have instead asked for an arbitrary bijection $\rho:N\rightarrow N$ such that $\rho(M)=W$ and required $f(\rho(\succsim))=\rho \ast f(\succsim)$ for all $\succsim$. In this case, since $\rho$ is a permutation, it is a member of the symmetric group. It therefore has a finite order $k$ so that $\rho^k$ is the identity. Since $\rho$ matches all men to women and vice-versa, $k$ is even. Hence $\rho^{k/2}$ is of order $2$. If $f(\rho(\succsim))=\rho \ast f(\succsim)$ for all $\succsim$ then $f(\rho^{m}(\succsim))=\rho^m \ast f(\succsim)$ for any $m$ and any $\succsim$. Thus $f(\rho^{k/2}(\succsim))=\rho^{k/2} \ast f(\succsim)$, and $f$ is weakly gender-neutral with respect to $\rho^{k/2}$.}

For a fixed symmetry $\sigma$ we say that a matching $\mu$ and respectively a profile of preferences $\succsim$ are \textbf{symmetric} if they equal to their reflections, so $\mu=\sigma \ast \mu$ and $\succsim=\sigma(\succsim)$. For notational simplicity, when considering a fixed weakly gender-neutral mechanism $f$ we will index the set of men and women $M=\{m_1,\dots, m_n\}$ and $W=\{w_1,\dots,w_n\}$ so that $\sigma(w_{j})=m_{j}$ for all $j$, where $\sigma$ is the symmetry with respect to which $f$ is weakly gender-neutral.

\begin{example}
Suppose $M=\{Ad,Bob, Carl\}$ and $W=\{Ann,Beth,Connie\}$. Let \\
$\sigma=\{(Ann,Ad),(Beth, Bob), (Connie, Carl)\}$.
The reflection of 
$$\mu=\{(Ann,Bob),(Beth,Carl),(Connie,Ad)\}$$
under $\sigma$ is the matching
$$\sigma \ast \mu=\{(Beth,Ad),(Connie,Bob),(Ann,Carl)\}$$
For instance $(Ann,Bob)\in \mu$ implies $(\sigma(Ann),\sigma(Bob))=(Ad,Beth)\in \sigma \ast \mu$. 
Fix a profile of preferences $\succsim$ such that $\succsim_{Ad}:Connie, Beth, Ann$ and such that  all agents other than Ad 
rank partners according to their alphabetical order. Then the reflected profile $\succsim'=\sigma(\succsim)$
is such that 
$\succsim'_{Ann}:Carl, Bob, Ad$ while all other agents rank all possible partners in alphabetical order.

\end{example}\label{example: weak gender-neutrality}

\begin{example}
The following is an example of a mechanism which is weakly gender-neutral under $\sigma(w_{j})=m_{j}$. For any preference profile, match $m_1$ and $w_1$ with their top choices unless they conflict -- i.e. if exactly one top ranks the other. In this case, match both with their top choices excluding one another. If $m_1$ and $w_1$ are matched together, repeat this with $m_2$ and $w_2$ choosing from the remaining agents. Continue until there is a pair $m_k$ and $w_k$ who do not choose one another. Suppose $m_k$ is matched with $w_j$ and $w_k$ is matched with $m_l$. If $j<l$ match the remaining agents using serial dictatorship with the men as dictators picking in order of their index. If $j>l$ match the remaining agents using serial dictatorship with the women as dictators picking in order of their index.
\end{example}

While $f$ is weakly gender-neutral in this example, men and women are treated highly unequally after most initial choices by $m_1$ and $w_1$: after most initial choices one side of the market retains all agency while the other is turned into objects. An agent who is uncertain about the royal's preferences might reason that they have an equal subjective likelihood to be dictator and object. For this agent, weak gender-neutrality might be sufficient. However, an agent with more information about the preferences of the royals might evaluate this mechanism as highly unfair. To avoid such unequal treatment we impose weak gender-neutrality on ``continuation'' submechanisms such as the mechanisms following on the choices by $m_1$ and $w_1$ in the present example.

\begin{definition}
A matching mechanism $g$ for the agents in $S$ is a \textbf{continuation submechanism} of the matching mechanism $f$ if there is a profile of preferences $\succsim_{-S}$ for $N\setminus S$ such that $g(\succsim_{S}^*)(j)=f(\succsim_S,\succsim_{-S})(j)\in S$ for all $\succsim_{S}$ and $j\in S$, where $\succsim_{S}^*$ is the restriction of $\succsim_{S}$ to $S$.
\end{definition}

\begin{definition}
A bilateral matching mechanism $f:\Omega^2\to \Sigma^2$ is \textbf{gender-neutral} if for every continuation submechanism $g$ of $f$ is weakly gender-neutral. 
\end{definition}

Since any mechanism is a continuation submechanism of itself, and gender-neutral mechanism is weakly gender-neutral. 

\begin{example}[Example \ref{example: weak gender-neutrality} continued]
To see that the mechanism $f$ defined above is not gender-neutral 
note that upon $m_1$ and $w_1$ respectively choosing $w_3$ and $m_2$ as their partners a continuation submechanism for all other agents arises. Since 
$w_1$'s partner ($m_2$) has a lower index than $m_1$'s ($w_3$) the remaining matches are determined by a serial dictatorship of all women. Since serial dictatorships are not not weakly gender-neutral, the overall mechanism is not gender-neutral. 
\end{example}

In the preceding example the 
continuation submechanism following the symmetric match between $m_1$ and $w_1$ is weakly gender-neutral under the same symmetry as the original mechanism. This observation easily generalizes: 

\begin{lemma}
Suppose $f$ is a weakly gender-neutral mechanism with symmetry $\sigma$. Suppose $g$ is a continuation submechanism of $f$. Say $g$ follows on a gender-neutral submatching. Then $g$ is weakly gender-neutral with the symmetry $\sigma'$ that is the restriction of $\sigma$ to the agents in $g$.
\end{lemma}

\section{Results and Approach}

We characterize the class of group strategy-proof, gender-neutral and efficient two-sided matching mechanisms. As a key lemma we characterize all group strategy-proof and efficient one-sided matching mechanisms, allowing agents to remain unmatched. For simplicity we ignore related special cases in both results, which we leave to the appendix. These correspond to when there are exactly four agents in two-sided matching and when there are exactly two agents in one-sided matching. These are special cases for similar reasons: both reduce to social choice problems with exactly two outcomes. As described in the introduction, our characterization results in the class of ``Royalty mechanisms." These mechanisms sequentially select two agents, one from either side, to choose their matches according to one of two regimes: matched-by-default(D) or unmatched-by-default(U). The precise order of these agents can vary. Royalty mechanisms are therefore parameterized by this order, which we call a ``succession order". 

\begin{definition} A \textbf{succession order} is a function $\varphi: A\rightarrow (M\times W)\times \{D,U\}$ where $A$ is a subset of $\Sigma^2_0$. Let $\varphi_1:A\rightarrow (M\times W)$ and $\varphi_2:A\rightarrow \{D,U\}$ correspond to the first and second components of $\varphi$. A succession order $\varphi$ must satisfy
\begin{enumerate}
    \item $\emptyset \in A$
    \item If $\nu\in A$ and $\varphi_1(\nu)=(m,w)$ then $m,w\notin N(\nu)$
    \item If $\nu\in A$ and $\varphi_1(\nu)=(m,w)$ then for any $m'.w'\notin N(\nu)$,  $\nu\cup\{(m,w'),(w,m')\}$ matches all but four or fewer agents in $A$. 
\end{enumerate}
\end{definition}

In addition to the succession order, we need to specify a terminal condition for the algorithm. Let $\chi$ be the set of group strategy-proof, efficient, and gender-neutral mechanisms with exactly four agents\footnote{These are characterized in appendix section \ref{appendix: four agents}.} and let $\Sigma_2^T$ be the set of submatchings in $\Sigma^2_0$ where exactly four agents are unmatched. 
\begin{definition}
A \textbf{terminal condition} is a map $\varphi_T:\Sigma_2^T\rightarrow \chi$.
\end{definition}

Given a succession order $\varphi$ and a terminal condition $\varphi_T$, the \textbf{royalty mechanism} $R_{(\varphi,\varphi_T)}$ chooses a matching using the following algorithm.

\begin{definition}
The \textbf{royalty algorithm} given a succession order $\varphi$, a terminal condition $\varphi_T$ and a preference profile $\succsim$ proceeds in a number of steps:\\
\begin{quote}
\textbf{Initialize:} Set $\nu_0=\emptyset$ and if there are three or more couples in the mechanism go to Step $1$, otherwise go to Step $T$\\
\textbf{Step $k$:}  The agents $\varphi_1(\nu_{k-1})$ are declared royals. If $\varphi_2(\nu_{k-1})=D$ then if either royal top-ranks the other among $N-N(\nu_{k-1})$, they are matched. If $\varphi_2(\nu_{k-1})=U$ the royals are matched only if they both top-rank one another among $N-N(\nu_{k-1})$. If the royals are not matched, each gets matched with their favorite partner excluding one another in $N-N(\nu_{k-1})$. Let $\nu_k$ denote the resulting submatching. If $\nu^k$ leaves $6$ or more couples unmatched, go to step $k+1$. If $\nu^k$ leaves 4 or fewer agents unmatched go to Step $T$. \\
\textbf{Step T:} If there are exactly two remaining agents, match them together to result in $\nu^T$. Otherwise use the mechanism $\varphi_T(\nu^k)$ to find a submatching $\overline{\nu}$ for the remaining four agents. Set $\nu^T = \nu^k\cup \overline{\nu}$. Return $\nu^T$.
\end{quote}
\end{definition}

The last step of royalty mechanisms are a special case which we discuss in the Appendix \ref{appendix: four agents}. There are a number of possible group strategy-proof, efficient and gender-neutral mechanisms. As a lead example consider a unanimity rule which sets one of the matchings as a default and chooses the other matching only if all four agents prefer it to the default.

\begin{theorem} \label{two-sided characterization}
A two-sided matching mechanism $f:\Omega^2\to \Sigma^2$ is group strategy-proof, efficient and gender-neutral only if it is a royalty mechanism.
\end{theorem}

Theorem \ref{two-sided characterization} does not characterize the set of all group strategy-proof, efficient and gender-neutral mechanisms as some royalty mechanisms are not gender-neutral. To see this, consider a royalty mechanism $f$ for 5 or more couples that starts with $m_1,w_1$ as the royal couple. Consider the cases that the royal couple either choose $w_2$ and $m_3$ as their partners or $w_3$ and $m_2$. Our definition of royalty mechanisms only imposes that these choices lead to either a matched-by-default or a unmatched-by-default step. However gender neutrality requires more than that: If $m_l$ and $w_k$ become the royal couple after the choice of $w_2$ and $m_3$ then $w_l=\sigma(m_l)$ and $m_k=\sigma(w_k)$ become must become the royal couple after the choice of $w_3$ and $m_2$.

To define neutral royalty mechanisms, we proceed inductively over the number of couples in a mechanism. To start say that the set of all royalty mechanisms for two couples consists of the  gender neutral, efficient and group strategy-proof mechanisms characterized in appendix \ref{appendix: four agents}. Now suppose that neutral royalty mechanisms for up to $n$ couples have been defined. To define a neutral royalty mechanism for $n+1$ couples, fix a symmetry $\sigma$ on $N$, with the standard assumption that $\sigma(m_i)=w_i$ for all $i$ and such that $(m_1,w_1)$ are chosen as the royal couple at the start of the mechanism.  

Say $(m_1,w_1)$ choose $m_k$ and $w_l$ as their partners. 
If these choices are symmetric with respect to $\sigma$, so if $k=l$, then a \textbf{neutral royalty mechanism} that is symmetric with respect to the restriction of $\sigma$ to the unmatched agents must be chosen as the continuation submechanism. If $k<l$ we may freely choose any neutral royalty mechanism $f^{k,l}$ to follow on the submatching $\{(m_1,w_l),(m_k,w_1)\}$. If $k>l$ we are bound by the choices for the preceding case. In that case we must use a permutation of the royalty mechanism $\sigma*f^{k,l}$ that is identical to $f^{k,l}$ except that wherever a man $m$ appears in $f^{k,l}$ the woman $\sigma(m)$ should take up his place in $\sigma*f^{k,l}$ and wherever a woman $w$ appears in $f^{k,l}$ the woman $\sigma(w)$ should take up her place in $\sigma*f^{k,l}$. 

\begin{corollary} \label{cor:two-sided characterization}
A two-sided matching mechanism $f:\Omega^2\to \Sigma^2$ is group strategy-proof, efficient and gender-neutral if only if it is a neutral royalty mechanism.
\end{corollary}

The proof is involved so we defer the details to the appendix. We sketch the argument here. The first key insight is to notice that within any gender-neutral two-sided mechanism is a one-sided mechanism. Suppose as usual that $\sigma(m_i)=w_i$ for all $i$. Let $\Omega^2_{symm}$ be the set of symmetric preference profiles (i.e. the $\succsim$ such that $\sigma(\succsim)=\succsim$). Since $\sigma(\succsim)=\succsim$, gender-neutrality implies that $f(\succsim)=\sigma \ast f(\succsim)$. Then if $(m_i,w_j)\in f(\succsim)$ so is $(\sigma(m_i),\sigma(w_j))=(w_i,m_j)$. From this we see that given any symmetric preference profile $\succsim$, for any pair $(m_i,w_i)$, either $m_i$ and $w_i$ are matched in $f(\succsim)$ or they swap partners with another symmetric pair $(m_j,w_j)$. By treating each symmetric pair $(m_i,w_i)$ as a single agent, $c_i$ we can extract a one-sided matching mechanism. We interpret the swap of partners between say $(m_i,w_i)$ and $(m_j,w_j)$ as a match between $c_i$ and $c_j$ and a match between $m_i$ and $w_i$ as the agent $c_i$ being unmatched. \citeasnoun{root2020incentives} showed that all group strategy-proof and efficient one-sided mechanisms are sequential dictatorships. However, they did not allow agents to remain unmatched, a key requirement for us. We complement their characterization by showing the same holds even if agents are allowed to remain unmatched.

\begin{definition}
let $N'=\{1,2,\dots,n\}$ and $N=\{m_1,\dots, m_n\}\cup\{w_1,\dots, w_n\}$. Given a gender-neutral mechanism $f$ for $N$, we define a one-sided matching mechanism $g$ for $N'$ which we call the \textbf{induced one-sided mechanism} for $f$. For any $\succsim$ in the one-sided matching market, let $\succsim^*$ be the preference profile in the two-sided market where $m_k\succsim^*_{w_j} m_l$ and $w_k\succsim^*_{m_j} w_l$ if and only if $k\succsim_{j} l$ for all $j,k,l$. Then for any $j$, $g(\succsim)(j)$ is the agent $i$ such that $f(\succsim^*)(m_j)=w_i$ and $f(\succsim^*)(w_j)=m_i$.
\end{definition} 

\begin{lemma}\label{lemma: induced roommates}
For any gender-neutral two-sided matching mechanism $f$, the induced one-sided mechanism $g$ is group strategy-proof and efficient if $f$ is.
\end{lemma}

\begin{proof}
Assume that $f$ is group strategy-proof and efficient. We will first show that $g$ is efficient. Consider any preference profile $\succsim$ in the one-sided matching market and the matching $g(\succsim)$. For any other one-sided matching $\mu\neq g(\succsim)$ we can define the symmetric two-sided matching $\nu$ so that $\nu(m_j)=w_i$ and $\nu(w_j)=m_i$ if and only if $\mu(j)=i$. Since $f$ is efficient, there must be some $j$ such that $f(\succsim^*)(m_j)\succ^{*}_{m_j}\nu(m_j)$. However by definition we then have that $g(\succsim)(j)\succ_{j}\nu(j)$. Hence $\nu$ cannot Pareto dominate $g(\succsim)$.

To see that $g$ is group strategy-proof, fix any preference profile $\succsim$, a group of agents $S$ and some $\succsim_S'$ such that $g(\succsim)\neq g(\succsim_S',\succsim_{-S})$.  Since $f$ is group strategy-proof, there is some $j$ among the set of men and women who have the same index as an agent in $S$ such that $f(\succsim^*)(m_j)\succ_{m_j} f((\succsim_S',\succsim_{-S})^*)(m_j)$ or $f(\succsim^*)(w_j)\succ_{w_j} f((\succsim_S',\succsim_{-S})^*)(w_j)$. By definition then $g(\succsim)\succ_j g(\succsim_S',\succsim_{-S})(j)$, so $\succsim_S'$ is not a profitable deviation for the coalition $S$. Since the coalition, deviation and original preference profile were arbitrary, $g$ is group strategy-proof.
\end{proof}

Before stating our characterization of one-sided matching mechanisms, let's recall the notion of sequential dictatorship.

\begin{definition} A \textbf{picking order} is a function $\phi: A\rightarrow N$ where $A$ is a subset of $\Sigma^1_0$ and 
\begin{enumerate}
    \item $\emptyset \in A$
    \item If $\nu\in A$ and $\phi(\nu)=i$ then $i\notin N(\nu)$
    \item If $\nu\in A$, $\phi(\nu)=i$ and $j\notin N(\nu)$ then 
     $\nu\cup\{(i)\}$ and $\nu\cup\{(i,j)\}$ are either in $A$ or are matchings. 
\end{enumerate}
\end{definition}

\begin{definition} A mechanism $f$ is a \textbf{sequential dictatorship} with respect to $\phi$ if for any $\succsim$, $f(\succsim)$ is the matching resulting from the following algorithm: \\
\begin{quote}
\textbf{Step $1$:} Agent $\phi(\emptyset)$ is matched with her favorite partner (including herself). Let $\nu_1$ be this submatching. \\
\textbf{Step $k\geq2$:} Agent $\phi(\nu_{k-1})$ is matched with her favorite remaining partner (including herself). Let $\nu_k$ be the matching $\nu_{k-1}\cup\{(\phi(\nu_{k-1})\}$ if $\phi(\nu_{k-1})$ prefers to remain unmatched or $\nu_{k-1}\cup\{(\phi(\nu_{k-1},j)\}$ if $j$ is $\phi(\nu_{k-1})$s favorite remaining partner. If $\nu_k$ is a matching, stop and return $\nu_k$.
\end{quote}
\end{definition}

\begin{theorem} \label{roommates mechanism}
A one-sided matching mechanism $f:\Omega^1\to \Sigma^1$ is group strategy-proof and efficient if and only if it is a sequential dictatorship.
\end{theorem}

We prove Theorem \ref{roommates mechanism} by induction over the number of agents. In a sequence of Lemmas we show that the result holds for 3 agents. For 4 agents we use the fact that restricted to the case where each agent bottom ranks being single, we face a classical social choice problem over 3 alternatives. This case must by the Gibbard-Satterthwaite Theorem be a dictatorship \cite{gibbard1973manipulation}\cite{satterthwaite1975strategy}. We show that the dictator in this constrained problem must also be the dictator in the four agent problem when agent may prefer to be single to some matches. The cases of five or more agents then do not require any special treatment. This result agrees with \citeasnoun{root2020incentives} who find a similar characterization in the setting where agents are not allowed to remain unmatched. Since agents are not able to veto the dictators' choices in a sequential dictatorship, Theorem \ref{roommates mechanism} admits an immediate corollary:

\begin{corollary}
There are no individually rational, group strategy-proof and efficient one-sided matching mechanisms
\end{corollary}


Having observed that a gender-neutral mechanism is equivalent to a one-sided matching mechanism on symmetric profiles, we see that on these profiles a group strategy-proof, gender-neutral and efficient mechanism must agree with some royalty mechanism. At a high level, the rest of the proof can be summarized by the following procedure. For any preference profile $\succsim$ we find a sequence of profiles $\succsim^0,\dots,\succsim^m$ where $\succsim^m=\succsim$ and $\succsim^0$ is symmetric. We then argue that for every pair of adjacent profiles $\succsim^k, \succsim^{k+1}$ in this sequence if a group strategy-proof, efficient and gender-neutral mechanism agrees with a royalty mechanism at $\succsim^k$ it must also do so at $\succsim^{k+1}$. The difficulty is in finding exactly the right sequence.

\section{The axioms}

No axiom can be dropped from the characterizations in Theorems \ref{two-sided characterization} and \ref{roommates mechanism}. Considering first the characterization of one-sided mechanisms in Theorem \ref{roommates mechanism} note that any constant mechanism is group strategy-proof but not efficient. 
Most efficient mechanisms are not group strategy-proof. For an example define a one-to-one function $R$ from the set of matchings to the natural numbers. Then map any profile $\succsim$ to the matching $\mu$ that minimizes $R(\mu)$ over the set of all efficient matchings. 

To see that this mechanism is not group strategy-proof fix a setting with just three agents $\{1,2,3\}$. Say $R$ assigns values 1,2,3 and 4 to the matchings $\{(1,2),(3)\}$, $\{(1),(2,3)\}$, $\{(1,3),(2)\}$ and $\{(1),(2),(3)\}$ respectively. Fix the profile $\succsim^*$ where each agent holds the same ranking $\succsim^*_i$ which ranks 3 above 1 above 2. Since each matching is efficient at $\succsim^*$, the mechanism chooses $\{(1,2),(3)\}$. Now modify agent 2's preference to be $\succsim'_2: 3,2,1$. Since now agents 1 and 2 would rather be single than be together the matching $\{(1,2),(3)\}$ is not efficient. Since  $\{(1),(2,3)\}$ is efficient and because $R(\{(1),(2,3)\})=2$ the mechanism maps   $(\succsim'_2,\succsim^*_{-2})$ to $\{(1),(2,3)\}$. However, agent 2 strictly prefers to be matched with agent 3 than to agent 1 so the mechanism is not strategy-proof.

Things get more interesting with Theorem \ref{two-sided characterization}. The constant mechanism which maps each profile of preferences to the matching $\sigma$ with respect to which $f$ is gender-neutral, is not only group strategy-proof but also gender-neutral. It is clearly not efficient. For a gender-neutral and efficient mechanism $f$ that is not group strategy-proof we fix, as above, a one-to-one mapping $R$ from the set of gender-neutral matchings to 
the natural numbers. To define $f(\succsim)$ first check whether there there exist any symmetric efficient matchings at $\succsim$. If so choose the  matching that minimizes $R(\mu)$ over the set of all efficient symmetric matchings. If not, use some fixed neutral royalty mechanism to calculate $f(\succsim)$. To see that the mechanism is gender-neutral fix a profile $\succsim$. Note that the set of efficient symmetric matchings at $\succsim$ and $\sigma(\succsim)$ coincide and that $\sigma\ast \mu$ holds for any symmetric matching. Therefore $\sigma \ast f(\succsim)= f(\succsim)=f(\sigma(\succsim))$ both equal the minimal $R(\mu)$ over the set of all matchings $\mu$ that are efficient at $\sigma$ (and therefore also at $\sigma(\succsim)$). Since neutral royalty mechanisms are gender neutral, $f(\succsim)=f(\sigma(\succsim))$ also holds for the alternative case where no gender neutral matchings are efficient at $\succsim$. The fact that $f$ is not group strategy-proof follows from the fact that the analog mechanisms for roommate problems is not group strategy-proof.\footnote{If $f$ was group strategy-proof it would have to be group strategy-proof on the subdomain of symmetric preferences. Now consider a problem with three couples. If the couples preferences and the function $R$ are given as in the discussion of roommate problems, we see that couple two has an incentive to lie about their preferences.}

If we drop gender-neutrality and modify efficiency to only include the preferences of one side of the market, we obtain the set of 
efficient and group strategy-proof one-sided matching mechanisms as characterized by \cite{pycia2017incentive} and \cite{bade2020random}.
When we replace gender-neutrality with weak gender-neutrality, our proof continues to hold until the point where we show that a royal couple must be matched or unmatched-by-default. However if this royal couple chooses partners $m$ and $w$ that are not symmetric, the continuation submechanism need not be symmetric. This is illustrated by example \ref{example: weak gender-neutrality}.

\subsection{The incompatibility of gender-neutrality and stability}

The tension between stability and incentive compatibility was first described in \citeasnoun{roth1982economics}. Here we show that stability also clashes with gender-neutrality. Stability therefore forces the designer to introduce asymmetry where none exists.

\begin{theorem}\label{theorem: stable never gender-neutral}
No stable mechanism $f$ for at least three couples is gender-neutral and stable.
\end{theorem}

\begin{proof}
Fix a matching problem with three couples $\{(m_1,w_1),(m_2,w_2),(m_3,w_3)\}$. Suppose $f$ was an stable mechanism that is gender-neutral with respect to $\sigma$. Fix the following preferences

\[
\begin{array}{c}
\succsim_{m_1} \\
\hline
w_3 \\
w_2 \\
w_1
\end{array}
\hspace{.2cm}
\begin{array}{c}
\succsim_{m_2} \\
\hline
w_1 \\
w_3 \\
w_2
\end{array}
\hspace{.2cm}
\begin{array}{c}
\succsim_{m_3} \\
\hline
w_2 \\
w_1 \\
w_3
\end{array}
\hspace{.5cm}
\begin{array}{c}
\succsim_{w_1} \\
\hline
m_3 \\
m_2 \\
m_1
\end{array}
\hspace{.2cm}
\begin{array}{c}
\succsim_{w_2} \\
\hline
m_1 \\
m_3 \\
m_2
\end{array}
\hspace{.2cm}
\begin{array}{c}
\succsim_{w_3} \\
\hline
m_2 \\
m_1 \\
m_3
\end{array}
\]
Notice that there are two stable matchings
\begin{align*} 
&M-\text{optimal }=\{ (m_1,w_3), (m_2,w_1), (m_3,w_2)\} \\
&W-\text{optimal }=\{ (m_1,w_2), (m_2,w_3), (m_3,w_1)\}
\end{align*}
To see that there is no other stable matching, notice that in any other matching at least one pair of agents with the same index $(m_i,w_i)$ are matched. In this case, there is an agent $m_j$ and an agent $w_j$ who top-rank $w_i$ and $m_i$ respectively. These form blocking pairs. 

Since $\sigma(\succsim)=\succsim$ and since $f$ is gender-neutral with respect to $\sigma$, $f(\succsim)$ must be symmetric. However neither stable matching is symmetric: $(m_1,w_3)$ are matched in the $M$-optimal stable match, but $(\sigma(m_1),\sigma(w_3))=(w_1,m_3)$ are not. Likewise $(m_1,w_2)$ are matched in the $W$-optimal stable match, but $(\sigma(m_1),\sigma(w_2))=(w_1,m_2)$ are not.
\end{proof}

\section{Two-sided Mechanisms with Randomization}

A typical approach to achieving fairness in social choice is to employ randomization. A deterministic mechanism is symmetrized by randomizing the role agents play in the mechanism. For example, one could symmetrize serial dictatorship, where agents are called in a fixed order to choose their preferred matches, by selecting a picking order uniformly at random. This gives a mechanism analogous to the mechanism known as random serial dictatorship (RSD) in house allocation problems. The symmetrized mechanism retains some of the features of the original mechanism. If the original mechanism is strategy-proof, so is the symmetrized version\footnote{In the sense that truthfully reporting gives a lottery which first-order stochastically dominates any other lottery that could be achieved by a misreport.}. On the other hand, there is no guarantee that the symmetrized mechanism will be efficient, even if the original mechanism is. For example, consider the following preferences:
\[
\begin{array}{c}
\succsim_{m_1} \\
\hline
w_3 \\
w_2 \\
w_1
\end{array}
\hspace{.2cm}
\begin{array}{c}
\succsim_{m_2} \\
\hline
w_1 \\
w_3 \\
w_2
\end{array}
\hspace{.2cm}
\begin{array}{c}
\succsim_{m_3} \\
\hline
w_2 \\
w_1 \\
w_3
\end{array}
\hspace{.5cm}
\begin{array}{c}
\succsim_{w_1} \\
\hline
m_3 \\
m_2 \\
m_1
\end{array}
\hspace{.2cm}
\begin{array}{c}
\succsim_{w_2} \\
\hline
m_1 \\
m_3 \\
m_2
\end{array}
\hspace{.2cm}
\begin{array}{c}
\succsim_{w_3} \\
\hline
m_2 \\
m_1 \\
m_3
\end{array}
\]
If we run serial dictatorship, selecting the picking order uniformly at random we get the following random allocation.
\begin{center}
\begin{tabular}{|c| c| c| c|} 
 \hline
  & $w_1$ & $w_2$ & $w_3$ \\
 \hline
 $m_1$ & $1/12$ & $11/24$ & $11/24$ \\ 
 \hline
 $m_2$ & $11/24$ & $1/12$ & $11/24$ \\
 \hline
 $m_3$ & $11/24$ & $11/24$ & $1/12$ \\
 \hline
\end{tabular}
\end{center}
However, the following random allocation gives a first-order stochastic improvement for all agents.
\begin{center}
\begin{tabular}{|c| c| c| c|} 
 \hline
  & $w_1$ & $w_2$ & $w_3$ \\
 \hline
 $m_1$ & $0$ & \hspace{0.18cm}$1/2$\hspace{0.18cm} & \hspace{0.2cm}$1/2$\hspace{0.18cm} \\ 
 \hline
 $m_2$ & \hspace{0.18cm}$1/2$\hspace{0.18cm} & $0$ & $1/2$ \\
 \hline
 $m_3$ & $1/2$ & $1/2$ & $0$ \\
 \hline
\end{tabular}
\end{center}
This is a common issue in randomized mechanisms in a variety of environments \cite{echenique2022efficiency}. In house allocation, \citeasnoun{bogomolnaia2001new} showed that it is inevitable: there is no efficient, strategy-proof and symmetric mechanism. There is a sense in which the mechanism is efficient: the outcome can be decomposed as a lottery over efficient deterministic outcomes. This is known as ex-post efficiency. \citeasnoun{bade2020random} showed that symmetrizing any one of the many efficient and group strategy-proof house allocation mechanism gives nothing other than random serial dictatorship. This, along with the impossibility result of \citeasnoun{bogomolnaia2001new} has lead to interest in the question of whether RSD is the unique ex-post efficient, strategy-proof and symmetric mechanism. 

Symmetrizing our royalty mechanisms gives a negative answer to the same question in two-sided matching. For example when symmetrizing a royalty mechanism where the royals are matched-by-default in each round we get the following allocation given the preferences above.

\begin{center}
\begin{tabular}{|c| c| c| c|} 
 \hline
  & $w_1$ & $w_2$ & $w_3$ \\
 \hline
 $m_1$ & \hspace{0.18cm}$1/9$\hspace{0.18cm} & \hspace{0.18cm}$4/9$\hspace{0.18cm} & \hspace{0.18cm}$4/9$\hspace{0.18cm} \\ 
 \hline
 $m_2$ & $4/9$ & $1/9$ & $4/9$ \\
 \hline
 $m_3$ & $4/9$ & $4/9$ & $1/9$ \\
 \hline
\end{tabular}
\end{center}

Notice that these three allocation matrices are Pareto ranked with the second matrix dominating RSD which in turn dominates the uniform match-by-default outcome. The ranking of RSD and uniform match-by-default is an artifact of this example; there are examples where the Pareto ranking is flipped and in general the two random allocations cannot be ranked.

{\footnotesize
\bibliography{matching}
\bibliographystyle{econometrica}
}

\pagebreak

\appendix

\section{The case of four agents in two-sided matching}\label{appendix: four agents}

In this section we characterize all group strategy-proof, efficient and gender-neutral bilateral matching mechanisms for four agents $\{m_1,m_2,w_1,w_2\}$. In this case $\Sigma^2$ contains only two matchings: $\nu\coloneqq\{(m_1,w_1),(m_2,w_2)\}$ and $\mu\coloneqq\{(m_1,w_2),(m_2,w_1)\}$. Given that there are only two matchings, each agent has exactly two preferences: the agent either prefers $\nu$ or $\mu$. A preference profile $\succsim$ can then be summarized as the set $S$ or all agents $i$ with $\nu(i)\succsim_i\mu(i)$. In the present context it is easier to represent preference profiles as such sets $S$, meaning that mechanisms now map from the collection of all subsets of $N$ to $\{\mu,\nu\}$.

In line with this representation of preference profiles a mechanism partitions the set of all subsets of $N$ into $\Lambda^\nu$ and $\Lambda^\mu$ with the understanding that $f(S)=\nu$ iff $S\in \Lambda^{\nu}$. Our three desiderata the translate to the following requirements in the environment with just two couples. 
\begin{itemize}
    \item The mechanism $f$ is efficient iff $N\in \Lambda^{\nu}$ and $\emptyset \in \Lambda^{\mu}$.
    \item The mechanism $f$ is group strategy-proof\footnote{In fact, group strategy-proofness is equivalent to individual strategy-proofness with two outcomes, so this condition is necessary and sufficient for the mechanism to be strategy-proof.} iff $S\in \Lambda^{\nu}$ and $S\subsetneq S'$ imply that $S'\in \Lambda^{\nu}$.
    \item The mechanism $f$ is gender-neutral if $S\in \Lambda^{\nu}$ implies that $\sigma(S)\in \Lambda^{\nu}$. 
\end{itemize}

Any group strategy-proof mechanism can therefore be represented by a set $\overline{\Lambda}^{\nu}\subset \Lambda^{\nu}$ of all minimal sets $S\in \Lambda^{\nu}$.

\begin{lemma} \label{four agents}
The following list of sets $\overline{\Lambda}^{\nu}$ is - up to renaming - an exhaustive list of all efficient group strategy-proof and gender-neutral mechanisms. 
\begin{itemize}
    \item[a)] $\{S\mid |S|=x\}$ for $x=1,2,3,4$,
    \item[b)] $\{\{m_1,w_1,w_2\}, \{m_1,w_1,m_2\},\{m_2,w_2\}\}$,
    \item[c)] $\{\{m_1,w_1,w_2\}, \{m_1,w_1,m_2\}\}$,
    \item[d)] Any subset of the sets $\{\{m_1,w_1\}\},\{\{m_2,w_2\}\},\{\{m_1,w_2\},\{m_2,w_1\}\},\{\{m_1,m_2\},\{w_1,w_2\}\} $.
\item[e)] $\{\{m_1\},\{w_1\},\{m_2,w_2\}\}$
\item[f)] $\{\{m_1\},\{w_1\}\}$
\end{itemize}
The 4 mechanisms listed in a) are moreover the exhaustive subset of fully symmetric mechanisms.
\end{lemma}

For notational convenience Lemma \ref{four agents} includes a certain number of double counting. The mechanism where $\nu$ is chosen if at least two agents prefer it for example falls in groups a) and d) above. We could have also economized the definition by noting that the mechanisms listed in b) and c) are equivalent to the mechanisms in e) and f) upon exchanging the matchings $\nu$ and $\mu$. The mechanism in f) is the matched-by-default mechanism with $m_1$ and $w_1$ the royal couple, the mechanism defined by $\{m_1,w_1\}$, which appears in d), is the unmatched-by-default mechanism with $m_1$ and $w_1$ as the royals.

\begin{proof}
For each of the mechanisms group strategy-proofness holds by definition since $\Lambda^{\nu}$ is in each case defined as the set of all supersets in $\overline{\Lambda}^{\nu}$. Each mechanism is Pareto optimal by the preceding argument together with the observation that $\emptyset\notin \overline{\Lambda}^{\nu}$ for any of the lists sets of subsets. Finally each mechanism is gender-neutral since each of the generating sets is gender-neutral. 

So we only have to show that the above list is exhaustive. To do so first consider $\overline{\Lambda}^{\nu}$ with $\{m_1,w_1,w_2\}\in \overline{\Lambda}^{\nu}$. By gender-neutrality  $\sigma(\{m_1,w_1,w_2\})=\{m_1,w_1,m_2\}\in \overline{\Lambda}^{\nu}$ must hold. Since  $\{m_1,w_1,w_2\},\{m_1,w_1,m_2\}\in \overline{\Lambda}^{\nu}$ no subset of either $\{m_1,w_1,w_2\}$ or $\{m_1,w_1,m_2\}$ can be in $\overline{\Lambda}^{\nu}$. So the latter cannot contain any of the singleton sets. $\{m_2,w_2\}$ is the only two agent set that can be contained in $\overline{\Lambda}^{\nu}$. We in sum found that the sets $\overline{\Lambda}^{\nu}$ which contain a three agent set are exactly the sets listed in a), b) and c). Exchanging $\nu$ and $\mu$ in the agents preferences we see that the sets $\overline{\Lambda}^{\nu}$ which contain a singleton are exactly the sets listed in  a)  e) and f). 

Finally to see that d) is an exhaustive list of mechanisms generated by two agent sets, note that d) lists the complete set of gender-neutral sets of two agent sets. 
\end{proof}

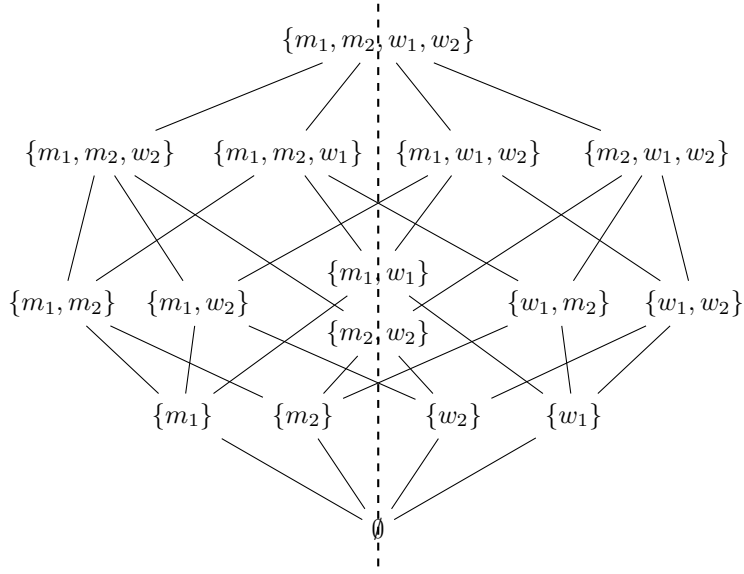
\begin{figure}
    \centering
    \begin{tikzpicture}
    \node(empty)      at (0,0)           {$\emptyset$};
    \node(m2)         at (-1,1.5)        {$\{m_2\}$};
    \node(m1)         at (-2.6,1.5)      {$\{m_1\}$};
    \node(w2)         at (1,1.5)         {$\{w_2\}$};
    \node(w1)         at (2.6,1.5)       {$\{w_1\}$};
    \node(m1w2)       at (-2.4,3)          {$\{m_1,w_2\}$};
    \node(m1m2)       at (-4.2,3)        {$\{m_1,m_2\}$};
    \node(w1m2)       at (2.4,3)           {$\{w_1,m_2\}$};
    \node(w1w2)       at (4.2,3)         {$\{w_1,w_2\}$};
    \node(m1w1)       at (0,3.4)         {$\{m_1,w_1\}$};
    \node(m2w2)       at (0,2.6)         {$\{m_2,w_2\}$};
    \node(-m2)        at (1.2,5)      {$\{m_1,w_1,w_2\}$};
    \node(-m1)        at (3.7,5)      {$\{m_2,w_1,w_2\}$};
    \node(-w2)        at (-1.2,5)       {$\{m_1,m_2,w_1\}$};
    \node(-w1)        at (-3.7,5)       {$\{m_1,m_2,w_2\}$};
    \node(N)          at (0,6.5)           {$\{m_1,m_2,w_1,w_2\}$};

    \draw(empty) -- (m1);
    \draw(empty) -- (m2);
    \draw(empty) -- (w1);
    \draw(empty) -- (w2);
    \draw(m1)    -- (m1m2);
    \draw(m1)    -- (m1w1);
    \draw(m1)    -- (m1w2);
    \draw(m2)    -- (m1m2);
    \draw(m2)    -- (w1m2);
    \draw(m2)    -- (m2w2);
    \draw(w2)    -- (m1w2);
    \draw(w2)    -- (w1w2);
    \draw(w2)    -- (m2w2);
    \draw(w1)    -- (m1w1);
    \draw(w1)    -- (w1w2);
    \draw(w1)    -- (w1m2);
    \draw(m1w2)  -- (-m2);
    \draw(m1w2)  -- (-w1);
    \draw(m1m2)  -- (-w1);
    \draw(m1m2)  -- (-w2);
    \draw(w1m2)  -- (-m1);
    \draw(w1m2)  -- (-w2);
    \draw(w1w2)  -- (-m1);
    \draw(w1w2)  -- (-m2);
    \draw(m1w1)  -- (-m2);
    \draw(m1w1)  -- (-w2);
    \draw(m2w2)  -- (-m1);
    \draw(m2w2)  -- (-w1);
    \draw(-m1)   -- (N);
    \draw(-m2)   -- (N);
    \draw(-w1)   -- (N);
    \draw(-w2)   -- (N);

    \draw [thick, dashed] (0,-0.5) -- (0,7);
    \end{tikzpicture}
    \caption{Four agents: Any gender-neutral, strategy-proof mechanism can be expressed as a set of nodes in the lattice above that is closed upwards and which is symmetric with respect to the the reflection over the vertical dotted line.}
    \label{lattice}
\end{figure}

\section{The case of two agents in one-sided matching}

Suppose that $N=\{1,2\}$. With just two agents, there are exactly two outcomes in one-sided matching: the matching $\{(1,2)\}$ and the matching $\{(1),(2)\}$.

\begin{definition}
We say that a mechanism $f$ is a \textbf{dictatorship} if there is an agent $k$ such that $f(\succsim)(k)=top(\succsim_k)$ for all $\succsim$.
\end{definition}

\begin{definition}
We say that a mechanism $f$ is a \textbf{unanimity rule} one of the two matchings is chosen unless both agree top-rank the other. 
\end{definition}

\begin{lemma}\label{lemma: two roomates}
If $N=\{1,2\}$, then
a one-sided mechanism $f:\Omega^1\to \Sigma^1$ is group strategy-proof and efficient if and only if it is either a dictatorship or a unanimity rule.
\end{lemma}

\begin{proof}
Since  $N=\{1,2\}$ there exist exactly two matchings:  one pairs the two agents, the other keeps both single. A mechanism is then efficient if and only if it chooses the matching preferred by both agents if they agree. The four such mechanisms map the two profiles where the agents disagree to the two different matchings. It is easy to check that these four mechanisms correspond to the two dictatorships and the two unanimity rules, and that they are strategy-proof. \end{proof}

\section{The proof of Theorem \ref{roommates mechanism}}

\subsection{Sequential Dictatorships are efficient and group strategy-proof.}

\begin{lemma}\label{lemma: sufficiency}
Sequential  Dictatorships are efficient and group strategy-proof.
\end{lemma}

\begin{proof}
Fix a sequential dictatorship $SD:\Omega^1\to \Sigma^1$.  Suppose $SD$ was not efficient or not group strategy-proof. So suppose there either exists some profile $\succsim$ and either a matching $\mu$ Pareto dominates $SD(\succsim)$ or there is a group $S\subset N$  and a deviation $\succsim'_S$ such that all members of $S$ weakly prefer $SD(\succsim'_S,\succsim_{-S})$ to $SD(\succsim)$ while some members of the group strictly hold this preference. Define $k^1$ and $k^2$ as the first rounds at which the $SD(\succsim)$-algorithm finds a match that differs from respectively $\mu$ and $SD(\succsim'_{S},\succsim_{-S})$. If only two agents remain unmatched at step $k^1$ or $k^2$ we obtain a contradiction via Lemma \ref{lemma: two roomates}. 

So say that at least three agents remain unmatched at steps $k^1$ and $k^2$. Say $i^1$ and $i^2$ are the dictators at these steps, so that $SD(\succsim)(i^1)\neq \mu(i^1)$ and $SD(\succsim)(i^2)\neq SD(\succsim'_{S},\succsim_{-S})(i^2)$.
Since the dictator $i^x$ gets matched with his top-ranked unmatched partner at Step $k^x$, and since $\mu(i^1)$ and $SD(\succsim'_{S},\succsim_{-S})(i^2)$ respectively stay available at Steps
 $k^1$ and $k^2$, we get the contradictions $SD(\succsim)(i^1)\succ_i \mu(i^1)$ and $SD(\succsim)(i^2)\succ_i SD(\succsim'_{S},\succsim_{-S})(i^2)$.
\end{proof}

\subsection{Submechanisms}

For any two different agents $j,k\in N$ define
$\Omega^{-j}$ and $\Omega^{-j,k}$ as the restriction of $\Omega^{1}$ to all agents but agent $j$ and $j,k$ respectively.\footnote{That is, $\Omega^{-j}$ is the set of preference profiles for all agents other than $j$ and such that $j$ is excluded from all other agents' preferences.  $\Omega^{-j,k}$ is similar except now $j$ and $k$ are excluded. }
For any group strategy-proof and efficient mechanism $f$, define two mechanisms $f^{-j}$ and $f^{-j,k}$ on $\Omega^{-j}$ and $\Omega^{-j,k}$ respectively as follows. Set $f^{-j}(\succsim^{-j})(i)=f(\succsim)(i)$ for all $\succsim^{-j}\in \Omega^{-j}$ and $i\in N\setminus \{j\}$ where $\succsim$ is any preference profile such that (1) when restricted to $N-\{j\}$ it gives $\succsim^{-j}$ and (2) all agents $i\neq j$ rank $j$ at the bottom while agent $j$ ranks herself at the top. Similarly let  $f^{-j,k}(\succsim^{-{j,k}})(i)=f(\succsim)(i)$ for all $\succsim^{-{j,k}}\in \Omega^{-j,k}$ and $i\in N\setminus \{j,k\}$ where $\succsim$ is any preference profile such that (1) when restricted to $N-\{j,k\}$ it gives $\succsim^{-j,k}$ and (2) all agents $i\neq j,k$ rank $j$ and $k$ at the bottom while agent $j$ and $k$ ranks each other at the top.
To simplify notation we drop the superscripts $-j$ and $-j,k$
from the preference profiles $\succsim^{-j}$ and $\succsim^{-j,k}$ in the sequel; $f^{-j}(\succsim)$ then stands for the application of $f^{-j}$ to the restriction of $\succsim$ to all agents but $j$.

\begin{lemma}\label{lemma: submechanisms}
Say $f$ is a group strategy-proof and efficient mechanism and $j,k\in N$ are two different agents. Then the
 mechanisms $f^{-j}$ and $f^{-j,k}$  are well defined, group strategy-proof and efficient.
 Moreover if $f(\succsim)(j)=j$, then $f(\succsim)(i)=f^{-j}(\succsim)(i)$ for all $i\in N\setminus \{j\}$ and if $f(\succsim)(j)=k$, then $f(\succsim)(i)=f^{-j,k}(\succsim)(i)$ for all $i\in N\setminus \{j,k\}$.
\end{lemma}

\begin{proof}
The arguments are similar for $f^{-j}$ and $f^{-j,k}$, so we will just prove this for $f^{-j}$.
To see that $f$ is well-defined, notice that for any $\succsim^{-j}$ and any profiles $\succsim$ and $\succsim'$ where the conditions described above hold (namely (1) when restricted to $N-\{j\}$ the profiles $\succsim$ and $\succsim'$  give $\succsim^{-j}$ and (2) all agents $i\neq j$ rank $j$ at the bottom while agent $j$ ranks herself at the top in both profiles) $\succsim$ and $\succsim'$ can only possibly differ in $j$'s ranking. However since $f$ is efficient, in both-cases $j$ will be matched to herself. By group strategy-proofness, all other agents matches must also remain the same.
Now suppose, by way of contradiction, that $f^{-j}$ is not group strategy-proof. Then there is a preference profile $\succsim^{-j}$, some coalition of agents $S$, and an profile $\succsim'_{S}$ for the agents in $S$ so that $f^{-j}(\succsim_S',\succsim^{-j}_{-S})(i)\succsim_i f^{-j}(\succsim^{-j})(i)$ for all $i$ in $S$ and for some $k$ in $S$, $f^{-j}(\succsim_S',\succsim^{-j}_{-S})(k)\succsim_k f^{-j}(\succsim^{-j})(k)$.
Let $\succsim_S^*$ be the profile for the agents in $S$ such that all agents bottom-rank $j$ and their ranking is the same as in $\succsim_S$ otherwise. Likewise, let $\succsim^{-j*}$ be the profile where all agents bottom-rank $j$ and their preferences are the same as $\succsim^{-j}$ otherwise. Let $\succsim_j$ be an arbitrary preference for $j$ where $j$ top-ranks herself. By definition $f^{-j}(\succsim^{-j})(i)=f(\succsim_j,\succsim^{-j*})(i)$ and $f^{-j}(\succsim_S',\succsim^{-j}_{-S})(i)=f(\succsim_S^*,\succsim_j,\succsim^{-j*}_{-S})(i)$ for all $i\neq j$. However this leads to a violation of group strategy-proofness for $f$, a contradiction. Finally, efficiency follows immediately from the efficiency of $f$.
\end{proof}

\subsection{The Inductive Structure of the Proof}

\begin{lemma}\label{lemma: how to do induction}
If each efficient and group strategy-proof mechanism $f:\Omega^1\to \Sigma^1$ for three or more agents has a dictator, then any efficient group strategy-proof mechanism is a sequential dictatorship. 
\end{lemma}

\begin{proof}
Fix a group strategy-proof and efficient mechanism $f:\Omega^1\to \Sigma^1$. If there are only two agents then $f$ is by Lemma \ref{lemma: two roomates} a sequential dictatorship. If there are more than two agents then $f$ has by assumption a dictator $i$, so that $f(\succsim)=\ptop(\succsim_i)$ for any $\succsim.$ If, upon matching $i$ with $\ptop(\succsim_i)$ only one agent remains unmatched, this agent must stay single. If not, then  Lemma \ref{lemma: submechanisms} implies that any choice of agent 1 is followed by a group strategy-proof and efficient submechanism $f'$.  By group strategy-proofness this submechanism $f'$ only depends on agent $i$'s choice (and not on  agent $i$'s rankings over the options he did not choose or on the preferences of the agent $j\neq i$ that $i$ did choose). 
 If two agents remain unmatched, $f'$ must by Lemma \ref{lemma: two roomates} be either a dictatorship or a unanimity rule, and we are done. If more than two agents remain unmatched $f'$ has by the assumption in the Lemma a dictator.  Proceeding inductively we reach the case where at most two agents remain. 
\end{proof}

Following Lemma \ref{lemma: how to do induction} it suffices to show that each group strategy-proof and efficient mechanism with more than two agents has a dictator. The next three section show  that each group strategy-proof and efficient mechanisms with respectively  $n=3$, $n=4$ and $n\geq 5$ agents has a dictator. 

\subsection{The case of three agents}
Throughout this section fix a group strategy-proof and efficient mechanism $f:\Omega^1\to \Sigma^1$ for the agents $N=\{1,2,3\}$, so that $\Sigma^1$ contains only four matchings. If any two agents are matched, the remaining agent is clearly left single.
The proof that $f$ must be a sequential dictatorship revolves around the notion of ``ownership''.
Say agent $i$ \textbf{owns} an agent $j$ if agent $i$ can always choose to be matched with agent $j$. Formally
$\ptop(\succsim_i)=j$ implies $f(\succsim_i,\succsim_{-i})(j)=i$ for all $\succsim_{-i}$. After establishing the  preliminary Lemma  \ref{lemma: enough for ownership} on ownership,
Lemma \ref{lemma: all agents are owned}  shows that each agent in $i\in \{1,2,3\}$ must be owned. Lemma \ref{lemma: one or two owners at the start} then rules out all ownership structures except then one where one agent owns all agents. 
Throughout this section, it will be useful to have some additional notation. Let $\succsim^{1,2}_i:1,2,3$, $\succsim^{1,3}_i:1,3,2$ and $\succsim^1_i$ as an arbitrary preferences with $\ptop(\succsim_i)=1$, so $\succsim^1_i$ equals either $\succsim^{1,2}_{i}$ or $\succsim^{1,3}_{i}$. Likewise  let $\succsim^{1,2}$, $\succsim^{1,3}$ and $\succsim^1$ denote preference profiles where all agents' preferences are  $\succsim^{1,2}_{i}$, $\succsim^{1,3}_{i}$ and $\succsim^1_{i}$ respectively.

The first lemma describes a sufficient condition to establish ownership.

\begin{lemma}\label{lemma: enough for ownership}
Agent $j^*$ owns agent $1$ if and only if
$f\big(\succsim^{1,2}\big)(j^*)=f\big(\succsim^{1,3}\big)(j^*)=1$.
\end{lemma}

\begin{proof}
If agent $j^*$ owns agent 1, then $f(\succsim)(j^*)=1$ holds for any $\succsim$ with $\ptop(\succsim_{j^*})=1$, in particular  $\succsim^{1,2}$ and $\succsim^{1,3}$.

\medskip

So suppose that we have $f\big(\succsim^{1,2}\big)(j^*)=f\big(\succsim^{1,3}\big)(j^*)=1$ for some agent $j^*$. For $j^*$ to own $1$ it is sufficient to show that 
$f(\succsim^1)(j^*)=1$ holds for all $\succsim^1$. To see that latter fix any $\succsim_{-j^*}$. Define $\succsim'_i$ for the two agents $i\neq j^*$ such that $\ptop(\succsim'_i)=1$ and  $2\succsim'_i3\Leftrightarrow 2\succsim_i3$. By the assumption $f(\succsim^1_{j^*},\succsim'_{-j^*})(j^*)=1$. Group strategy-proofness then yields  $f(\succsim^1_{j^*},\succsim'_{-j^*})=f(\succsim^1_{j^*},\succsim_{-j^*})$, so that $f(\succsim^1_{j^*},\succsim_{-j^*})(j^*)=1$.

\medskip

Case 1: $j^*=1$.

\medskip

Case 1.1:  $f\big(\succsim^{1,2}\big)(2)=2$. 
Then the group strategy-proofness of $f$ implies 
$f\big(\succsim^{1,2}_2, \succsim^{1}_{-2}\big)(1)=1$. The strategy-proofness of $f$ implies that $f\big(\succsim^{1,3}_2,\succsim^{1}_{-2}\big)(2)\in \{2,3\}$, 
and therefore, by either group strategy-proofness or feasibility respectively, $f\big(\succsim^{1,3}_2,\succsim^1_{-2}\big)(1)=1$.
We in sum get that $f(\succsim^1)(j^*)=1$  for any $\succsim^1$, which establishes the claim. 

Case 1.2: $f\big(\succsim^{1,3}\big)(3)=3$. Applying the arguments from Case 1.1. mutatis mutandis we get $f(\succsim^1)(j^*)=1$  for any $\succsim^1$.

Case 1.3: $f\big(\succsim^{1,2}\big)(3)=f\big(\succsim^{1,3}\big)(3)=2$. So 
$f\big(\succsim^{1,2}\big)=f\big(\succsim^{1,3}\big)$ and by group strategy-proofness $f$ must also match $3$ with $2$ for any profile where agent 1 ranks being single at the top and where at least one agent in $\{2,3\}$ prefers being matched with the other to being single. For the last remaining case where agent $1$ top ranks being single while the other two prefer being single to being matched with each other start with the observation that $f(\succsim^1_1,\succsim^{1,2}_{-1})$ and $f(\succsim^1_1,\succsim^{1,3}_{-1})$ both match agents 2 and 3. By strategy-proofness we then get $f(\succsim^1_1,\succsim^{1,2}_2\succsim^{1,3}_3)(2)\in \{2,3\}$ as well as 
$f(\succsim^1_1,\succsim^{1,2}_2\succsim^{1,3}_3)(3)\in \{2,3\}$. The efficiency of $f$ then implies that agents 2 and 3 both stay single in $f(\succsim^1_1,\succsim^{1,2}_2\succsim^{1,3}_3)$, so that once again $f(\succsim^1_1,\succsim^{1,2}_2\succsim^{1,3}_3)(1)=1$ and in sum $f(\succsim^1)(1)=1$. 
\medskip

Case 2: $j^*\neq 1$. W.o.l.g say $j^*=2$ so that $f\big(\succsim^{1,2}\big)(2)=f\big(\succsim^{1,3}\big)(2)=1$. 
Now say $f(\succsim^*)(2)\neq 1$ did hold for some
$\succsim^*$ with $\ptop(\succsim^*_1)(i)=1$ for $i=1,2,3$. The assumption $f\big(\succsim^{1,2}\big)(2)=f\big(\succsim^{1,3}\big)(2)=1$ together with group strategy-proofness implies that $f(\succsim^{1,2}_1,\succsim^1_{-1})(1)=2$ (A) as well as $f(\succsim^{1,3}_3,\succsim^1_{-3})(1)=2$ (B) for any $\succsim^1$. So for $f(\succsim^*)(1)\neq 2$ to hold we must have $\succsim^*_1=\succsim^{1,3}_1$ and $\succsim^*_3=\succsim^{1,2}_3$.  Case 1: $f(\succsim^*)(1)= f(\succsim^{1,3}_1,\succsim^*_2,\succsim^{1,2}_3)(1)=3$. In that case we get by group strategy-proofness that $f(\succsim^{1,3}_1,\succsim^*_2,\succsim^{1,2}_3)=f(\succsim^{1,3}_1,\succsim^*_2,\succsim^{1,3}_3)$ which leads to a contradiction to (B) established above.  Case 1: $f(\succsim^*)(1)= f(\succsim^{1,3}_1,\succsim^*_2,\succsim^{1,2}_3)(1)=1$. In that case we get by group strategy-proofness that $f(\succsim^{1,3}_1,\succsim^*_2,\succsim^{1,2}_3)=f(\succsim^{1,2}_1,\succsim^*_2,\succsim^{1,2}_3)$ which leads to a contradiction to (A) established above.

\end{proof}

\begin{lemma}\label{lemma: all agents are owned}
According to $f$, each agent $i\in\{1,2,3\}$ is owned.
\end{lemma}

\begin{proof}
 By Lemma \ref{lemma: enough for ownership} (and the interchangeability of all agents) it suffices to show 
 $f(\succsim^{1,2})(1)=f(\succsim^{1,3})(1)$. So suppose we had $f(\succsim^{1,2})(1)\neq f(\succsim^{1,3})(1)$.

\medskip

Case 1: Agent 1 is not alone  at either $\succsim^{1,3}$ or $\succsim^{1,2}$, so $f(\succsim^{1,2})(1)\neq 1\neq f(\succsim^{1,3})(1)$.

\medskip
Suppose we had $f(\succsim^{1,2})(3)=f(\succsim^{1,3})(2)=1$. The group strategy-proofness  of $f$ implies $f(\succsim^{1,2})=f(\succsim^{1,3}_3,\succsim^{1,2}_{-3})$ as well as $f(\succsim^{1,3})=f(\succsim^{1,2}_2,\succsim^{1,3}_{-2})= f(\succsim^{1,3}_3,\succsim^{1,2}_{-3})$. 
We in sum get the contradiction $f(\succsim^{1,2})= f(\succsim^{1,3})$  to the assumption that $f(\succsim^{1,2})$ and $f(\succsim^{1,3})$ respectively match agent 1 with agent 2 and 3.
So we must have $f(\succsim^{1,2})(2)=f(1,3)(3)=1$.

\bigskip
The proof derives a contradiction by showing that $f$ must equal two different matchings at the profile $\succsim$ given by
\begin{eqnarray*}
\succsim_1:  1,3,2\\
\succsim_2: 3, 1, 2\\
\succsim_3: 1,2,3
\end{eqnarray*}

Starting at $f\big(\succsim^{1,2}\big)(1)=2$  change agent 2's preference to $\succsim_2: 3,1,2$.
  By strategy-proofness $f\big(\succsim_2,\succsim^{1,2}_{-2}\big)(2)\in\{1,3\}$. Since the matching $\{(1),(2,3)\}$ Pareto dominates the matching $\{(1,2),3\}$ at $(\succsim_2,\succsim^{1,2}_{-2})$,  $f\big(\succsim_2,\succsim^{1,2}_{-2}\big)$ must match agents $2$ and $3$. The group srategy-proofness then implies that  $f\big(\succsim_2,\succsim^{1,2}_{-2}\big)=f(\succsim)$. On the other hand, 
the group strategy-proofness of $f$, together with the observation that $\succsim^{1,3}_1=\succsim_1$ implies that $f\big(\succsim^{1,3}\big)=f\big(\succsim\big)$. A contradiction arises  since $f\big(\succsim^{1,3}\big)(1)=3$ while we have shown above that  $f\big(\succsim\big)(1)=1$.

\bigskip

Case 2: Agent 1 is alone at $\succsim^{1,2}$ or $\succsim^{1,3}$, not both. W.o.l.g $f(\succsim^{1,2})(1)=1$.

\medskip

Case 2.1.  $f\big(\succsim^{1,2}\big)(2)=2$.  
Group strategy-proofness then implies  $f\big(\succsim^{1,2}_2, \succsim^{1,3}_{-2}\big)(1)=1$. Strategy-proofness yields that  $f\big(\succsim^{1,3}\big)(2)$ equals 2 or 3. We then get $f\big(\succsim^{1,3}\big)(1)=1$ (in the first case by group strategy-proofness and in the second by feasibility)
a  contradiction to the assumption that agent 1 is not alone at $\succsim^{1,3}$.

Case 2.2 $f\big(\succsim^{1,2}\big)(2)=3$.
Group strategy-proofness then implies   $f\big(\succsim^{1,2}\big)=f\big(\succsim^{1,2}_3,\succsim^{1,3}_{-3}\big)$. Since $f\big(\succsim^{1,3}\big)\neq f\big(\succsim^{1,2}\big)= f\big(\succsim^{1,2}_3,\succsim^{1,3}_{-3}\big)$ and since
 $\succsim^{1,2}_3$ and $\succsim^{1,3}_3$ differ only in their ranking of agents 2 and 3 in second and third place,  $f\big(\succsim^{1,3}\big)(3)=3$. So for $f\big(\succsim^{1,3}\big)(1)\neq 1$ to hold we must have $f\big(\succsim^{1,3}\big)(1)= 2$.

\medskip

Just as above the
 proof derives a contradiction by showing that $f$ maps the following $\succsim$ to two different matchings,
\begin{eqnarray*}
\succsim_1:  3,1,2\\
\succsim_2: 1, 2,3\\
\succsim_3: 1,3,2.
\end{eqnarray*}

Strategy-proofness, $\succsim^{1,2}_2=\succsim_2$, and $f\big(\succsim^{1,2}\big)(1)=1$ imply
  $f\big(\succsim^{1,2}_{3},\succsim_{-3}\big)(1)\in\{1,3\}$. If $f\big(\succsim^{1,2}_{3},\succsim_{-3}\big)(1)=1$, then the group strategy-proofness of $f$ implies $f\big(\succsim^{1,2}\big)=f\big(\succsim^{1,2}_{3},\succsim_{-3}\big)$. A contradiction arises, since $f\big(\succsim^{1,2}\big)$ is at $\big(\succsim^{1,2}_{3},\succsim_{-3}\big)$ dominated by $\mu$ with $\mu(1)=3$. So
  $f\big(\succsim^{1,2}_{3},\succsim_{-3}\big)(1)=3$ and  $f\big(\succsim^{1,2}_3,\succsim_{-3}\big)=\mu$. Since $f$ is strateyproof and since $\mu(3)=1=\ptop(\succsim_3)= \ptop(\succsim^{1,2}_3)$ we get $f\big(\succsim\big)=\mu$.
  Starting with $f\big(\succsim^{1,3}\big)$ and noting that  $\succsim^{1,3}_3=\succsim_3$ the group strategy-proofness of $f$ yields the contradiction $ f\big(\succsim\big)=f\big(\succsim^{1,3}\big)\neq \mu$. 

\end{proof}

\begin{lemma}\label{lemma: one or two owners at the start}
The mechanism $f$ must have a dictator.
\end{lemma}

\begin{proof}
By  Lemma \ref{lemma: all agents are owned} each agent is owned. Suppose we had an ``ownership chain'' in the sense that agent $i$ owns $j$ who in turn owns $k$ with $\{1,2,3\}=\{i,j,k\}$. For a profile $\succsim$ with $\ptop(\succsim_i)=j$ and $\ptop(\succsim_j)=k$ we would then obtain the contradiction $(i,j),(j,k)\in f(\succsim)$. Given that there can be no such ownership chains
we have to consider only 3 ownership structures (up to renaming): 1. Each agent owns herself. 2.  Agent 1 owns agent 2 and agent 3 owns herself. 3. One agent owns all three agents. To show that the third case must hold we rule out the first two.

\bigskip

Case 1: \textbf{Each agent owns herself.}
The classic example of a roommate problems  without a stable matching serves to obtain a contradiction. Define $\succsim$ as follows:
\begin{eqnarray*}
\succsim_1: 2,3,1\\
\succsim_2: 3,1, 2\\
\succsim_3: 1,2,3
\end{eqnarray*}

Since $f$ is efficient at least two agents must get matched. Assume that $f(\succsim)(1)=2$  and consider the deviation to $\succsim'_2: 3,2,1$ and
 $\succsim'_3: 2,3,1$ for agents 2 and 3. Since each agent owns herself, $f(\succsim_1,\succsim'_{-1})(2)\succsim_2 2$ and $f(\succsim_1,\succsim'_{-1})(3)\succsim_3 3$. So $f(\succsim_1,\succsim'_{-1})$ either keeps all agents single or pairs up agents 2 and 3. Since the latter Pareto dominates the former, $f(\succsim_1,\succsim'_{-1})(2)=3$ must hold. Since $f(\succsim_1,\succsim'_{-1})(2)=3\succ_2 1=f(\succsim)(2)$ and  $f(\succsim_1,\succsim'_{-1})(3)=2\succ_3 3=f(\succsim)(3)$ a contradiction to the group strategy-proofness of $f$ results.
 Mutatis mutandis the same arguments rule out the remaining two matchings. So it cannot be that each agent owns herself.

\bigskip

Case 2: \textbf{Agent 1 owns agent 2 and agent 3 owns herself.}
Transform the profile $\succsim$ in two steps to $(\succsim_2,\succsim'_{-2})$ where
\begin{eqnarray*}
\succsim_1: 3,2,1&&\succsim'_1:3,1,2\\
\succsim_2: 3,2, 1&&\\
\succsim_3: 2,1,3&&\succsim'_3: 2,3,1.
\end{eqnarray*}
Since agent 1 owns agent 2, $f(\succsim)(1)$ equals 2 or 3. The latter must hold since  the matching $\{(1,2),(£)\}$  is (at $\succsim$) Pareto dominated by $\{(1,3)(2)\}$. 
By group strategy-proofness $f(\succsim)=f(\succsim'_1,\succsim_{-1})$.  
Now swap agent 3's preference to $\succsim'_3$: By strategy-proofness and since agent 3 owns herself $2\succ_3 f(\succsim_2,\succsim'_{-2})(3)\succsim'_3 3$, so that $f(\succsim_2,\succsim'_{-2})(3)=3$. Conditioning on agent 3 staying single, 
 agents 1 and 2 must, by efficiency, also stay single at $f(\succsim_2,\succsim'_{-2})$. A contradiction arises since matching agents 2 and 3 Pareto dominates $f(\succsim_2,\succsim'_{-2})$ at $(\succsim_2,\succsim'_{-2})$.
\end{proof}

\subsection{The case of four agents}
In the present subsection assume that $f:\Omega^{1}\to \Sigma^{1}$ is a group strategy-proof and efficient mechanism for four agents.
For each profile of preferences $\succsim$ define $\overline{\succsim}$ so that each agent $i$ ranks being single at the bottom keeping all other rankings identical to $\succsim$.
Say $\overline{\Omega}^1$ is the subdomain of such profiles where each agent ranks being single at the bottom. The set of matchings where no agent is single is $\overline{\Sigma}^1$. There are exactly four such matchings.

\begin{lemma}\label{lemma: 1 dictates overline}
The restriction $\overline{f}$ of $f$ to $\overline{\Omega}^1$ is a dictatorship.
\end{lemma}

\begin{proof}
Since $f$ is group strategy-proof and efficient, its restriction to $\overline{\Omega}^1$ is so too. Since $f$ is efficient and since each agent ranks being single at the bottom no agent stays single according to $f(\succsim)$ for any $\succsim\in \overline{\Omega}^1$ and we can represent $\overline{f}$ as a mechanism mapping $\overline{\Omega}^1$ to $\overline{\Sigma}^1$. Since $\overline{\Sigma}^1$ contains  three matchings, each of which is fully determined by the match of a single agent, we are facing a  classic social choice problem with three options where four agents may hold any preferences over these three options. By the Gibbard Satterthwaite theorem $\overline{f}:\overline{\Omega}^1\to\overline{\Sigma}^1$ must be a dictatorship.
\end{proof}

For the reminder of the present section say agent 1 is the dictator in $\overline{f}:\overline{\Omega}^1\to\overline{\Sigma}^1$.

\begin{lemma}
For each $i\in\{1,2,3,4\}$, $f^{-i}$ is a sequential dictatorship.
\end{lemma}

\begin{proof}
By Lemma \ref{lemma: submechanisms} $f^{-i}$ is efficient and group strategy-proof for three agents. By the preceding section $f^{-i}$ is a sequential dictatorship.
\end{proof}

For the remainder of the section assume that agent 2 is the dictator in $f^{-1}$.

\begin{lemma}\label{lemma: same first dict in m-3 m-4}
 An agent $i^*\in\{1,2\}$  is the dictator in both $f^{-3}$ and $f^{-4}$.
\end{lemma}

\begin{proof}
First fix a profile $\succsim$ such that
\begin{eqnarray*}
&\succsim_1:&1,\cdot,\cdot,\cdot\\
&\succsim_2:&2,\cdot,\cdot,\cdot\\
&\succsim_3:&2,3,4,1\\
&\succsim_4:&2,4,3,1
\end{eqnarray*}

The efficiency of $f$  implies $f(\succsim)(1)=1$, and Lemma \ref{lemma: submechanisms} then implies $f^{-1}(\succsim)\subset f(\succsim)$. Since $\ptop(\succsim_2)=2$ and since agent 2 is the dictator in $f^{-1}$, $f(\succsim)(2)=f^{-1}(\succsim)(2)=2$. By efficiency  agents 3 and 4 also remain single.  By Lemma \ref{lemma: submechanisms} $f(\succsim)$ must then be consistent with $f^{-3}(\succsim)$ and $f^{-4}(\succsim)$.
 Since $f(\succsim)(4)=4$ and since $2\succsim_4 4$, agent 4 cannot be the dictator in $f^{-3}$. Mutatis mutandis we see that agent 3 cannot be the dictator in $f^{-4}$.

 To see that $f^{-3}$ and $f^{-4}$ must have the same dictator fix $\succsim$ such that
 \begin{eqnarray*}
&\succsim_1:&1,2,\cdot,\cdot,\\
&\succsim_2:&1,2,\cdot,\cdot,\\
&\succsim_3:&3,\cdot, \cdot, \cdot,\\
&\succsim_4:&4,\cdot, \cdot, \cdot.
\end{eqnarray*}
  Since $f$ is efficient, $f(\succsim)(3)=3$ and $f(\succsim)(4)=4$. By Lemma \ref{lemma: submechanisms},  $f(\succsim)(1)=f^{-3}(\succsim)(1)=f^{-4}(\succsim)(1)$. 
  Since $\ptop(\succsim_1)=\ptop(\succsim_2)=1$, we then get $f^{-3}(\succsim)(1)=f^{-4}(\succsim)(1)=i^*$ so  $f^{-3}$ and $f^{-4}$ have the same dictator. \end{proof}

\begin{lemma}\label{lemma: dictators in the submechanisms}
Agent 1 is the dictator in $f^{-j}$ for $j=2,3,4$.
\end{lemma}

\begin{proof}
Fix $\succsim$ and the deviations $\succsim'$
\begin{eqnarray*}
\succsim_1: 3,1,2,4&\succsim'_1: 3,4,1,2\\
\succsim_2:3,2,1,4&\\
\succsim_3: 2,1,3,4&\\
\succsim_4: 4,\cdot,\cdot,\cdot\hspace{0.2cm}&\succsim'_4: 1,4,\cdot,\cdot\\
\end{eqnarray*}
By efficiency $f(\succsim)(4)=4$. By Lemma \ref{lemma: submechanisms} $f^{-4}(\succsim)\subset f(\succsim)$. Suppose 1 is not the dictator in $f^{-4}$. So 2 or 3 must be the dictator in $f^{-4}$ 
and we get $f(\succsim)=\{\{1\},\{2,3\},\{4\}\}$.  
 Changing agent 1 and 4's preferences to
 $\succsim'_{1,4}$,  group strategy-proofness implies $f(\succsim'_{1,4},\succsim_{2,3})(1)=4$. By  efficiency  $f(\succsim'_{1,4},\succsim_{2,3})(2)=3$. Now define $\succsim''$ to be identical to $(\succsim'_{1,4},\succsim_{2,3})$ except that each agent ranks being single at the bottom. By group strategy-proofness $f(\succsim'_{1,4},\succsim_{2,3})=f(\succsim'')$. A contradiction arises since $\succsim''\in \overline{\Sigma}^1$ and agent 1 is by assumption the dictator for problems where all agents rank being single at the bottom. So agent 1 must be the dictator in $f^{-4}$. 
 By Lemma \ref{lemma: same first dict in m-3 m-4} agent 1 is also the dictator in $f^{-3}$.

 Now define $\succsim$
 \begin{eqnarray*}
&\succsim_1:& 1,3,\cdot,\cdot,\\
&\succsim_2:& 2,\cdot, \cdot,\cdot\\
&\succsim_3:& 1,3,\cdot, \cdot,\\
&\succsim_4:& 4,\cdot, \cdot, \cdot.\\
\end{eqnarray*}

By efficiency $f(\succsim)(2)=2$ and $f(\succsim)(4)=4$. By
   Lemma \ref{lemma: submechanisms}  $f(\succsim)$ is consistent with
   $f^{-2}(\succsim)$  and $f^{-4}(\succsim)$. Since agent 1 is the dictator in $f^{-4}$, $f^{-4}(\succsim)(1)=1=f(\succsim)(1)=f^{-2}(\succsim)(1)$. The latter implies that
   agent 3 cannot be the dictator in
     $f^{-2}$. Mutatis mutandis the dictator of $f^{-2}$ cannot be agent 4 either. So the dictator of $f^{-2}$ must be agent 1.
\end{proof}

\begin{lemma}
Agent 1 is the dictator in $f$.
\end{lemma}

\begin{proof}

We have to show that $f(\succsim)(1)=\ptop(\succsim_1)$ for all $\succsim$.
Fix an arbitrary $\succsim$. Assume without loss of generality that $\ptop(\succsim_1)\in\{1,2\}$. If $f(\succsim)(k)=k$ for $k=3$ or $k=4$, then $f^{-k}(\succsim)\subset f(\succsim)$. Since 1 is the dictator in $f^{-k}$ for $k=3,4$ we get that $f^{-k}(\succsim)(1)=f(\succsim)(1)=\ptop(\succsim_1)$. So for the remainder assume that $(3,4)\in f(\succsim)$.

\medskip
\textbf{Case 1: $ \ptop(\succsim_1)=2$.}
Suppose we have $f(\succsim)(1)\neq 2$.

Since $(3,4)\in f(\succsim)$,  $f(\succsim)$ must then equal $\{(1),(2),(3,4)\}$. By group strategy-proofness we can w.l.o.g assume that

 \begin{eqnarray*}
&\succsim_1:& 2,1,3,\cdot,\\
&\succsim_2:& 2,\cdot, \cdot,\cdot\\
&\succsim_3:& 4,\cdot,\cdot, \cdot,\\
&\succsim_4:& 3,4\cdot, \cdot, \cdot.\\
\end{eqnarray*}

By strategy-proofness $f(\succsim'_1,\succsim_{-1})(1)\in \{1,3\}$ for  $\succsim'_1: 2,3,1$. If $f(\succsim'_1,\succsim_{-1})(1)=1$, then group strategy-proofness implies $f(\succsim'_1,\succsim_{-1})=f(\succsim)$, and $f(\succsim'_1,\succsim_{-1})$ is consistent with $f^{-2}(\succsim'_1,\succsim_{-1})$. A contradiction arises since agent 1, would as the dictator in $f^{-2}$ choose agent 3. So we must have $f(\succsim'_1,\succsim_{-1})(1)=3$. Efficiency then implies $f(\succsim'_1,\succsim_{-1})(4)=4$ so that $f(\succsim'_1,\succsim_{-1})$ is consistent with $f^{-4}(\succsim'_1,\succsim_{-1})$.  A contradiction arises since agent 1, would as the dictator in $f^{-4}$ choose agent 2.

\textbf{Case 2: $\ptop(\succsim_1)=1$.} 
Since $(3,4)\in f(\succsim)$,  $f(\succsim)$ must then equal $\{(1,2),(3,4)\}$. . By group strategy-proofness assume w.l.o.g. that  $\succsim$ is given by

 \begin{eqnarray*}
&\succsim_1:& 1,2,3,\cdot,\\
&\succsim_2:& 1,2\cdot, \cdot,\\
&\succsim_3:& 4,\cdot,\cdot, \cdot,\\
&\succsim_4:& 3,4\cdot, \cdot,.\\
\end{eqnarray*}
Now swap agents 2 and 3 in agent 1's ranking, so that $\succsim'_1:1,3,2$. By strategy-proofness, $f(\succsim'_1,\succsim_{-1})(1)\in \{2,3\}$. Suppose  $f(\succsim'_1,\succsim_{-1})= 2$, then we have $f(\succsim'_1,\succsim_{-1})=f(\succsim)$ by group strategy-proofness. Now define $\succsim''$ to be identical to $(\succsim'_1,\succsim_{-1})$ except that each agent drops being single to the bottom of their ranking. By group strategy-proofness $f(\succsim'')=f(\succsim'_1,\succsim_{-1})$. A contradiction arises since $\succsim''\in \overline{\Sigma}^1$, but $f(\succsim'')(1)=2\neq \ptop(\succsim'')(1)=3$. So we must have $f(\succsim'_1,\succsim_{-1})(1)=3$. By efficiency agents 2 and 4 must then stay single. A contradiction arises since $f(\succsim'_1,\succsim_{-1})(1)$ must then equal $f^{-2}(\succsim'_1,\succsim_{-1})(1)=1$ since agent 1 is the dictator in $f^{-2}$. 

\end{proof}

\subsection{The case of $n\geq 4$ agents}
For this section fix a group strategy-proof and efficient mechanism $f$ for $n+1$ agents and assume that all such mechanisms for $n$ or fewer agents are sequential dictatorships.

\begin{lemma}\label{lemma: more than 4 all submech, same dict}
If agent 1 is the dictator in  $f^{-2}$, then some agent $j\in\{1,2\}$ is the dictator in $f^{-k}$ for all $k\notin \{1,2\}$.
\end{lemma}

\begin{proof}
Suppose agent $4$ was the dictator in $f^{-3}$.
Fix $\succsim$ such that each agent $i\neq 3$  ranks agents 1 and 3 respectively at the top and  at the bottom, agents $i\neq 1,3$ rank being single  in second place and $\ptop(\succsim_3)=3$.

\begin{eqnarray*}
&\succsim_1: &1,\cdot, \cdot, \cdots, 3\\
&\succsim_2: &1, 2,\cdot,  \cdots, 3\\
&\succsim_3: &3, \cdot, \cdot, \cdots\\
&\succsim_4: &1, 4,\cdot, \cdots, 3\\
&\cdot&\\
&\cdot&
\end{eqnarray*}

By efficiency $f(\succsim)$ is consistent with $f^{-3}(\succsim)$. Since agent 4 is the dictator in $f^{-3}$, and since $\ptop(\succsim_4)=1$, we have $f(\succsim)(4)=1$. By efficiency $f(\succsim)(i)=i$ for all $i\notin \{1,4\}$, in particular $f(\succsim)(2)=2$, so that $f^{-2}(\succsim)$ is consistent with $f(\succsim)$. Since agent 1 is the dictator in $f^{-2}$, we obtain the contradiction that $f(\succsim)(1)=f^{-2}(\succsim)(1)=1$. Since agents 3 and 4 were chosen arbitrarily in $N\setminus \{1,2\}$ the dictator in $f^{-k}$ is for each $k\notin\{1,2\}$ is either agent 1 or agent 2.

To see that $f^{-k}$ for each $k\notin \{1,2\}$ must have the same dictator consider the profile $\succsim$ where agents 1 and 2 both rank agents 1 and 2 in first and second place and where any agent $i\notin \{1,2\}$ top ranks being single, so

\begin{eqnarray*}
&\succsim_1:& 1,2, \cdot, \cdot,\\
&\succsim_2:& 1, 2,\cdot, \cdot, \cdot,\\
&\succsim_3:& 3, \cdot, \cdot, \cdot,\\
&\succsim_4:& 4,\cdot, \cdot, \cdot,\\
&\cdot&\\
&\cdot&
\end{eqnarray*}

By efficiency, $f(\succsim)(i)=i$ for each $i\notin \{1,2\}$. So $f(\succsim)$ is consistent with $f^{-i}(\succsim)$ for each $i\notin \{1,2\}$. By the preceding paragraph some agent $j\in \{1,2\}$ is the dictator in $j^{-3}$, so $f(\succsim)(j)=f^{-3}(\succsim)(j)=1$. Now consider any $k\notin\{1,2,3\}$. Since $f(\succsim)$ is consistent with $f^{-k}$, agent $j$ must also be the dictator in $f^{-k}$

\end{proof}

By Lemma \ref{lemma: more than 4 all submech, same dict} one agent $j$ is the dictator in all submechanisms $f^{-k}$ with $j\neq k$. For the remainder assume w.l.o.g that agent 1 is this agent.

\begin{lemma}\label{lemma: all j-k submechs have the same dictator}
Agent 1 is the dictator in $f^{-j,k}$ for any two agents $j,k$ such that $1\notin \{j,k\}$.
\end{lemma}

\begin{proof}
Fix $\succsim$ such that two agents other than agent 1, say agents 2 and 3 top rank each other while a third agent, say agent 4, top ranks himself. All remaining agents top rank agent 1 and rank agents 2, 3 and 4 at the bottom.
\begin{eqnarray*}
&\succsim_1:& 1, \cdot, \dots, 2,3,4\\
&\succsim_2:& 3, \cdot, \dots, \\
&\succsim_3:& 2, \cdot, \dots, \\
&\succsim_4:& 4, \cdot, \dots, \\
&\succsim_5: &1, \cdot, \dots, 2,3,4\\
&\cdot&\\
&\cdot&\\
&\cdot&
\end{eqnarray*}

 By efficiency $f(\succsim)$ matches agent 4 with himself and agents 2 and 3 with each other. So $f(\succsim)$ is consistent with $f^{-4}(\succsim)$ and $f^{-2,3}(\succsim)$. Since agent 1 is the dictator in $f^{-4}(\succsim)$, $f(\succsim)(1)=1$. Since $f^{-2,3}(\succsim)$ is consistent with $f(\succsim)$, $f^{-2,3}(\succsim)(1)=1$. Since all agents but agent 4 top rank agent 1 in $\succsim$, the dictator in $f^{-2,3}$ is either agent 1 or agent 4. Repeating the same arguments with swapping agent 4 and 5, we see that either agent 1 or agent 5 is the dictator in $f^{-2,3}$. In sum, agent 1 must be the dictator in $f^{-2,3}$. Since agents 2 and 3 were chosen arbitrarily, agent 1 is the dictator in any $f^{-j,k}$ for $1\notin \{j,k\}$.
\end{proof}

\begin{lemma}
Agent 1 is the dictator in $f$.
\end{lemma}

\begin{proof}
Fix a profile $\succsim$. Suppose $f(\succsim)(1)\neq \ptop(\succsim_1)$.
Since $n\geq 5$, we can fix an agent $j$ such that  $\{j, f(\succsim)(j)\}\cap \{1, f(\succsim)(1), \ptop(\succsim_1)\}=\emptyset$.
If $j=f(\succsim)(j)$, then $f(\succsim)$ is consistent with $f^{-j}(\succsim)$. Since
agent 1 is by Lemma \ref{lemma: more than 4 all submech, same dict} the dictator in $f^{-j}$ and since $\ptop(\succsim_1)\neq j$ we then obtain $f(\succsim)(1)=f^{-j}(\succsim)(1)=\ptop(\succsim_1)$. If $j\neq f(\succsim)(j)$, define $f(\succsim)(j)=k$. In that case $f(\succsim)$ is consistent with $f^{-j,k}(\succsim)$.
Since agent 1 is by Lemma \ref{lemma: all j-k submechs have the same dictator} the dictator in $f^{-j,k}$ and since $\ptop(\succsim_1)\notin \{j,k\}$ we can conclude as above that $f(\succsim)(1)=f^{j,k}(\succsim)(1)=\ptop(\succsim_1)$.
\end{proof}

\section{Proof of Theorem \ref{two-sided characterization}}

Fix a mechanism $f:\Omega^{2}\rightarrow \Sigma^{2}$. Restricted to the domain  $\Omega^{2}_{symm}$ the mechanism $f$ is by Lemma \ref{lemma: induced roommates} and Theorem \ref{roommates mechanism} a serial dictatorship in which one couple can choose to either stay alone (marry each other) or pair up with a different couple (swap partners). 
Recall that $f$ is assumed to be weakly gender-neutral with respect to $\sigma$ that is defined such that $\sigma(m_i)=w_i$ for all $i$. 
Without loss of generality say the couple  $(m_1,w_1)$ is the dictator in the embedded roommates mechanism. 
We call this the ``royal couple" and all other agents ``commoners." The rest of this proof is dedicated to showing that the royal couples' powers also apply to symmetric preferences. At the same time we must show that any conflict between the royal couples interests (when one wants the royals to marry and the other does not) must be mediated using either the matched-by-default or the unmatched-by-default protocol.

For Lemmas \ref{lemma: at least one gets best}, \ref{lemma: wanting commoners and royals}, and \ref{lemma: when possible both get best} fix a symmetric profile $\succsim^{sym}_{-m_1,w_1}$ for all commoners. 

\begin{lemma}\label{lemma: at least one gets best}
 Say $\Omega^{l,k}_{m_1,w_1}$ is the set of all preferences of $m_1$ and $w_1$ that respectively top rank $w_l$ and $m_k$. Then either
\begin{itemize}
 \item $f(\succsim_{m_1,w_1},\succsim^{sym}_{-m_1,w_1} )(m_1)=w_l$ and $f(\succsim_{m_1,w_1},\succsim^{sym}_{-m_1,w_1} )(w_1)=m_k$ 
 for all $\succsim_{m_1,w_1}\in \Omega^{l,k}_{m_1,w_1}$, or

 \item $f(\succsim_{m_1,w_1},\succsim^{sym}_{-m_1,w_1} )(m_1)\neq w_l$ and $f(\succsim_{m_1,w_1},\succsim^{sym}_{-m_1,w_1} )(w_1)=m_k$
 for all $\succsim_{m_1,w_1}\in \Omega^{l.k}_{m_1,w_1}$ or
 \item $f(\succsim_{m_1,w_1},\succsim^{sym}_{-m_1,w_1} )(m_1)=w_l$ and $f(\succsim_{m_1,w_1},\succsim^{sym}_{-m_1,w_1} )(w_1)\neq m_k$
 for all $\succsim_{m_1,w_1}\in \Omega^{l,k}_{m_1,w_1}$.
 \end{itemize}
\end{lemma}

\begin{proof}

Pick an arbitrary $\succsim_{m_1,w_1}\in \Omega^{l,k}_{m_1,w_1}$.

\textbf{Case 1: $k=l$}  
Let $\succsim'_{m_1}=\sigma(\succsim_{w_1})$  so that
Theorem \ref{roommates mechanism} implies $(m_1,w_k),(w_1,m_k)\in f(\succsim'_{m_1},\succsim_{w_1},\succsim^{sym}_{-m_1,w_1})$ Since $f(\succsim'_{m_1},\succsim_{w_1},\succsim^{sym}_{-m_1,w_1})(m_1)=\ptop(\succsim'_{m_1})=\ptop(\succsim_{m_1})$, group strategy-proofness  then implies that $(m_1,w_k),(w_1,m_k)\in f(\succsim_{m_1,w_1},\succsim^{sym}_{-m_1,w_1})$.

\textbf{Case 2: $k\neq l$}

Let $\succsim'_{m_1}: w_l, w_k$ and $\succsim'_{w_1}: m_k,m_l$. By group strategy-proofness and Case 1, $f(\succsim'_{m_1,w_1},\succsim^{sym}_{-m_1,w_1})$ either marries the royal couple with their most preferred partners (\textbf{Case 2.1}), or with $m_l$ and $w_l$ (\textbf{Case 2.2}) or with $m_k$ and $w_k$ (\textbf{Case 2.3}). In Case 2.1 the group strategy-proofness of $f$ implies $f(\succsim_{m_1,w_1},\succsim^{sym}_{-m_1,w_1})$ matches each royal with their top choice.

For Case 2.2 
Strategy-proofness implies $f(\succsim_{w_1},\succsim'_{m_1},\succsim^{sym}_{-m_1,m_1})(w_1)\neq m_k$.
Now since $m_1$ could swap to make an announcement symmetric to $\succsim_{w_1}$, strategy-proofness and Case 1 imply that $f(\succsim_{w_1},\succsim'_{m_1},\succsim^{sym}_{-m_1,m_1})(m_1)\succsim'_{m_1} w_k$. However, if $f(\succsim_{w_1},\succsim'_{m_1},\succsim^{sym}_{-m_1,m_1})(m_1)=w_k$, we obtain a violation of group strategy-proofness since if $w_1$ announces $\sigma(\succsim_{m_1})$ she gets the $\succsim_{w_1}$-preferred $m_k$. So $f(\succsim_{w_1},\succsim'_{m_1},\succsim^{sym}_{-m_1,m_1})(m_1)=w_l$ must hold. Finally, by strategy-proofness for $m_1$ we have $f(\succsim_{w_1},\succsim'_{m_1},\succsim^{sym}_{-m_1,m_1})(m_1)=
f(\succsim_{m_1,w_1},\succsim^{sym}_{-m_1,m_1})(m_1)=w_l$. Group strategy-proofness then implies  $f(\succsim_{w_1},\succsim'_{m_1},\succsim^{sym}_{-m_1,m_1})=
f(\succsim_{m_1,w_1},\succsim^{sym}_{-m_1,m_1})$, so that $$f(\succsim_{w_1},\succsim'_{m_1},\succsim^{sym}_{-m_1,m_1}(m_1)=
f(\succsim_{m_1,w_1},\succsim^{sym}_{-m_1,m_1})(w_1)\neq m_k$$ as required.
 
 Mutatis mutandis the arguments of Case 2.2 apply to Case 2.3. 
\end{proof}

\begin{lemma}\label{lemma: wanting commoners and royals}
 Say $\Omega^*_{m_1,w_1}$ is the set of preferences for the royals where exactly one royal top ranks the other. 
Then either a) or b) holds for all $\succsim_{m_1,w_1}\in \Omega^*_{m_1,w_1}$ 

\begin{itemize}
    \item[a)] $(m_1,w_1)\in f(\succsim_{m_1,w_1},\succsim^{sym}_{-m_1,m_1})$ .
    \item[b)] $f(\succsim_{m_1,w_1},\succsim^{sym}_{-m_1,m_1})$ matches the royal who top ranks a commoner with that commoner. 
\end{itemize}
\end{lemma}

\begin{proof}
Let $\succsim'_{m_1}:w_1,w_k$ and $\succsim'_{w_1}:m_k, m_1$ for $k\neq 1$. Since it is not possible to match both royals with their top-ranked partners, 
Lemma \ref{lemma: at least one gets best} yields that one of the royals  gets their top partner. 

\textbf{Case 1: $(m_1,w_1)\in f(\succsim'_{m_1,w_1},\succsim^{sym}_{-m_1,m_1})$.} Then Lemma \ref{lemma: at least one gets best} implies that $(m_1,w_1)\in f(\succsim''_{m_1,w_1},\succsim^{sym}_{-m_1,w_1})$ for any $\succsim''_{m_1,w_1}$ with $\ptop(\succsim''_{m_1})=w_1$ and $\ptop(\succsim''_{w_1})=m_k$ (including $\succsim'_{w_1}$ that rank $m_1$ last). Then group strategy-proofness implies that $(m_1,w_1)\in f(\succsim_{m_1,w_1},\succsim^{sym}_{-m_1,w_1})$ for any $\succsim_{m_1,w_1}$ with $\ptop(\succsim_{m_1})=w_1$. By gender-neutrality we also get that $(m_1,w_1)\in f(\succsim_{m_1,w_1},\succsim^{sym}_{-m_1,w_1})$ if $\ptop(\succsim_{w_1})=m_1$ and if $\succsim_{m_1}$ is any preference.

\textbf{Case 2: $f(\succsim'_{m_1,w_1},\succsim^{sym}_{-m_1,m_1})(w_1)= m_k$.} If $(m_1,w_1)\in f(\succsim'_{m_1,w_1},\succsim^{sym}_{-m_1,m_1})$ held for some $\succsim_{m_1,w_1}$ with
 $\ptop(\succsim_{m_1})=w_1$ and $\ptop(\succsim_{w_1})=m_l$ for $l\neq 1$ case 1 would imply the contradiction that $(m_1,w_1)\in f(\succsim'_{m_1,w_1},\succsim^{sym}_{-m_1,m_1}) $. 
 Hence by Lemma \ref{lemma: at least one gets best} we have $(w_1,m_l) \in f(\succsim_{(m_1,w_1)},\succsim^{sym}_{-m_1,w_1})(w_1)=m_l$. gender-neutrality then implies the result.  
\end{proof}

\begin{lemma}\label{lemma: when possible both get best}
Say $\ptop(\succsim_{m_1})=w_l$ and $\ptop(\succsim_{w_1})=m_k$. If $k=1=l$ or if $k\neq 1\neq l$,
$f(\succsim_{m_1,w_1},\succsim^{sym}_{-m_1,w_1})$ matches the royal couple with their most preferred partners.
\end{lemma}

\begin{proof}
If $k=1=l$ the claim follows from first part of the proof of Lemma \ref{lemma: at least one gets best}. So let $k\neq 1 \neq l$. From Lemma \ref{lemma: wanting commoners and royals} there are two cases to consider.

\textbf{Case 1: If exactly one royal top-ranks the other, the royals are matched} Consider $\succsim'_{m_1}: w_l, w_1$ and $\succsim'_{w_1}: m_k, m_1$. By Lemma \ref{lemma: at least one gets best} at least one royal must get their top partner. By strategy-proofness, neither royal can do worse than the other royal. Since $m_1$ is only matched with $w_1$ if $w_1$ is matched with $m_1$ we then get  
 $(w_1,m_k),(m_1,w_l)\in f(\succsim'_{m_1,w_1},\succsim^{sym}_{-m_1,w_1})$ and by group strategy-proofness $f(\succsim_{m_1,w_1},\succsim^{sym}_{-m_1,w_1})=f(\succsim'_{m_1,w_1},\succsim^{sym}_{-m_1,w_1})$.

\textbf{Case 2: If exactly one royal top-ranks the other,  the royal top-ranking a commoner gets their top match} Consider the preferences 

\begin{eqnarray*}
\succsim'_{m_1}: w_l,w_1,&& \succsim''_{m_1}: w_1,w_l\\
\succsim'_{w_1}: m_k,m_1,&& \succsim''_{w_1}: m_1,m_k
\end{eqnarray*}
derived from $\succsim_{m_1}$ and $\succsim_{w_1}$ keeping all else equal.

Since we are in Case 2, we have $f(\succsim'_{m_1},\succsim''_{w_1},\succsim^{sym}_{-m_1,w_1})(m_1)=w_l$ and $f(\succsim'_{w_1},\succsim''_{m_1},\succsim^{sym}_{-m_1,w_1})(w_1)=m_k$. 
By group strategy-proofness we then get that $(m_1,w_l),(w_1,m_k)\in f(\succsim'_{m_1},\succsim'_{w_1},\succsim^{sym}_{-m_1,w_1})(m_1).$ 
Applying group strategy-proofness once again to drop each other in their royals' rankings we get $(m_1,w_l),(w_1,m_k)\in f(\succsim_{m_1,w_1},\succsim^{sym}_{-m_1,w_1})$ as required.  

\end{proof}

The preceding three Lemmas referred to an arbitrarily fixed symmetric profile for all commoners $\succsim^{sym}_{-m_1,w_1}$. We showed that for any such fixed profile, the royals must be matched according to a matched-by-default or a unmatched-by-default protocol. The next Lemma shows that for a vast set of profiles for the commoners the royals get their top choices, if these top choices do not stand in conflict (so if they both want to marry commoners or want to marry each other.)

\begin{lemma}\label{lemma: Kew Gardens}
Fix $\succsim$ such that a subset of all pairs of commoners have symmetric preferences, while all remaining commoners bottom-rank the royals, but do not necessarily have symmetric preferences. Say $\ptop(\succsim_{w_1})=m_k$ and $\ptop(\succsim_{m_1})=w_l$. If  $l=k=1$ or  $l\neq 1\neq k$, then $(m_1,w_l), (w_1,m_k)\in f(\succsim)$. 
\end{lemma}

\begin{proof}
 We use induction over the number $m$ of pairs who do not have symmetric preferences. 

\medskip

\textbf{Start of the induction:} $m=0$ so that $\succsim_{-m_1,w_1}$ is symmetric. In this case the claim holds by Lemma \ref{lemma: when possible both get best}.

\medskip

\textbf{Induction step:} Suppose the claim holds up to some $m<n$. Fix an arbitrary profile $\succsim_{-m_1,w_1}$ such that $m+1$ pairs have (potentially) asymmetric preferences, ranking the 
royals at the bottom, and such that the remaining $n-m-1$ pairs have symmetric preferences. 
Without loss, suppose that $(m_2,w_2)$ are a type that do not have symmetric preferences and consider $\succsim'_{w_1}: m_k,m_2$.

\textbf{Case 1: $k\neq 1 \neq l$.} Suppose by way of contradiction that $f(\succsim)(w_1)\neq m_k$.

Case 1.1: Couple $(m_k,w_k)$ does not have symmetric preferences. For $\succsim'_{m_k}=\sigma(\succsim_{w_k})$, $\pbottom(\succsim_{w_k})=m_1$ implies $\pbottom(\succsim'_{m_k})=w_1$. Since $(\succsim'_{m_k},\succsim_{-m_k})$ is covered by the hypothesis of the induction $f(\succsim'_{m_k},\succsim_{-m_k})(w_1)=m_k$. We then obtain a contradiction to strategy-proofness since $w_1\neq f(\succsim)(m_k)\succ'_{m_k}f(\succsim'_{m_k},\succsim_{-m_k})(m_k)=w_1=\pbottom(\succsim'_{m_k})$.

Case 1.2 Couple $(m_k,w_k)$ does have symmetric preferences.  By strategy-proofness and the assumption that $f(\succsim)(w_1)\neq m_k$, we have $f(\succsim'_{w_1},\succsim_{-w_1})(w_1)\neq m_k$. By Case 1.1, $w_1$ would get matched with $m_2$ if she were to top rank $m_2$ (since we've assumed $m_2$ and $w_2$ announce asymmetric preferences). So by strategy-proofness $f(\succsim'_{w_1},\succsim_{-w_1})(w_1)= m_2$. Now consider $\succsim'_{m_2}=\sigma(\succsim_{w_2})$. Since couple $(m_2,w_2)$ does not have symmetric preferences $\pbottom(\succsim_{m_2})=w_1$. We then obtain a contradiction to strategy-proofness since $(\succsim'_{w_1,m_2},\succsim_{-w_1,m_2})$ is covered by the hypothesis of the induction so that $f(\succsim'_{w_1,m_2},\succsim_{-w_1,m_2})(w_1)=m_k$. The latter implies that $m_2$ can improve his match by switching his preference from $\succsim_{m_2}$ to $\succsim'_{m_2}$ to avoid being matched with $w_1=\pbottom(\succsim_{m_2})$.  

\medskip

Cases 1.1 and 1.2 proves that $(w_1,m_k)\in f(\succsim)$ holds in Case 1. By gender-neutrality $(m_1,w_l)\in f(\succsim)$ also holds.

\medskip

\textbf{Case 2: k=l=1.} Suppose $(m_1,w_1)\notin f(\succsim)$. As in the proof of Case 1.2 
$f(\succsim'_{w_1},\succsim_{-w_1})(w_1)=m_2$. Just as in that proof we obtain a contradiction to strategy-proofness since $m_2$ would be better-off if he symmetrized his preference with $w_2$.

\end{proof}

The next two Lemmas pertain to the case that the royals top rank a symmetric pair $(m_i,w_i)$. In this case the royals are matched with their top choices - for any profile of commoners.  

\begin{lemma}
Fix $\succsim$ so that the royals top rank each other. Then $f(\succsim)$ marries the royals.
\end{lemma}

\begin{proof}
Suppose not. Define $\succsim'_{-m_1,w_1}$ as a modification of $\succsim_{-m_1,w_1}$ where $m_1$ and $w_1$ are moved to the top of all agents preferences, but are otherwise unchanged. Sequentially swap each commoner $i$'s preference from $\succsim_i$ to $\succsim'_i$. By group strategy-proofness, with each such swap the matching either stays constant or commoner $i$ marries a royal. Therefore $(m_1,w_1)\notin f(\succsim_{m_1,w_1},\succsim'_{-m_1,w_1})$. Suppose $(m_1,w_i), (w_1,m_j)\in f(\succsim_{m_1,w_1},\succsim'_{-m_1,w_1})$, with $i\neq 1 \neq j$ and  $i=j$  permitted. 
Start with $\succsim'_{-m_1,w_1}$, to define a new profile $\succsim''_{-m_1,w_1}$ by dropping the royals to the bottom of all 
 commoners' rankings all commoners other than in $w_{i}, m_{i}, w_{j}$ and $m_{j}$ rankings. Moreover let $\succsim'_{w_i}\colon=\sigma(\succsim_{m_i})$
 and if $i\neq j$ also $\succsim'_{m_j}\colon=\sigma(\succsim_{w_j})$ 
 so that the couples $(m_i,w_i)$ as well as $(m_j,w_j)$ (if different) have symmetric preferences where the agents who are matched with the royal couples by $f(\succsim_{m_1,w_1},\succsim'_{-m_1,w_1})$ top rank the royal couple. 
  By group strategy-proofness none of these changes affect the match so that $f(\succsim_{m_1,w_1},\succsim'_{-m_1,w_1})= f(\succsim_{m_1,w_1},\succsim''_{-m_1,w_1})$. However, a contradiction arises since $\succsim''_{-m_1,w_1}$ is covered by Lemma \ref{lemma: Kew Gardens}, so  $(m_1,w_1)\in f(\succsim_{m_1,w_1},\succsim''_{-m_1,w_1})$.

\end{proof}

\begin{lemma}\label{lemma: the royals may choose the same partners}
If $\ptop(\succsim_{m_1})=w_l$ and $\ptop(\succsim_{w_1})=m_l$, then $(m_1,w_l),(w_1,m_l)\in f(\succsim)$.
\end{lemma}

\begin{proof}
Fix the royal's preferences such that they rank each other in second place. Suppose  $ f(\succsim)$ did not match both royals with their top ranked partners.

\medskip
\textbf{Observation 1:} One royal must get their most preferred partner. 

By the preceding Lemma and group strategy-proofness the royals cannot both prefer each other to their matches $f(\succsim)(m_1)$ and $f(\succsim)(w_1)$. Now suppose we had $f(\succsim)(w_1)=m_1$. By group strategy-proofness  $f(\succsim)=f(\succsim_{m_1,w_1},\succsim'_{-m_1,w_1})$ where the commoners' preferences $\succsim'_{-m_1,w_1}$ are derived from $\succsim_{-m_1,w_1}$ by dropping the royals to the bottom of each commoner's preference keeping all else equal. A contradiction arises since $\succsim'_{-m_1,w_1}$ is covered by  Lemma \ref{lemma: Kew Gardens}, so that $(m_1,w_1)\in f(\succsim_{m_1,w_1},\succsim'_{-m_1,w_1})$. In the only remaining case one royal gets their most preferred partner. 

\medskip

\textbf{Observation 2:} We may w.l.o.g assume that $\pbottom(\succsim_{m_l})=w_1$ and $\pbottom(\succsim_{w_l})=m_1$.

Suppose this did not hold. For concreteness assume that $(m_1,w_l)\in f(\succsim)$, noting that the following arguments apply mutatis mutandis to the alternative case where $(w_1,m_l)\in f(\succsim)$. Since $(w_1,m_l)\notin f(\succsim)$, we may 
by group strategy-proofness  w.l.o.g assume that $\pbottom(\succsim_{m_l})=w_1$. Now define $\succsim'_{w_l}$ by dropping $m_1$ to the bottom of $w_l$'s ranking keeping all else equal to $\succsim_{w_l}$. If $(m_1,w_l)\notin f(\succsim'_{w_l},\succsim_{-w_l})$, then Observation 1 implies that $(w_1,m_l)\in f(\succsim'_{w_l},\succsim_{-w_l})$. 

\medskip

By Observations 1 and 2 and group strategy-proofness we may assume that $(m_1,w_l)(w_1,m_j)\in f(\succsim)$  for some $j\neq l$, $\ptop(\succsim_{m_j})= w_1$, $\pbottom(\succsim_{m_i})=w_1$ for all $i\neq j,1$ and  $\pbottom(\succsim_{w_i})=m_1$ for all $i\neq 1$

\medskip
 Now lift $m_1$ in $w_j$'s ranking so that  $\ptop(\succsim'_{w_j})=m_1$ keeping all else equal to $\succsim_{w_j}$.
 By group strategy-proofness we either have $f(\succsim)=f(\succsim'_{w_j},\succsim_{-w_j})$ or $(m_1,w_j)\in f(\succsim'_{w_j},\succsim_{-w_j})$. 
 Now define $\succsim''_{m_j,w_j}$ as two gender-neutral preferences $\succsim''_{m_j}=\sigma(\succsim''_{w_j})$ so that an agent in $\{m_j,w_j\}$ who is matched with a commoner keeps their preference. By group strategy-proofness $f(\succsim)=f(\succsim'_{w_j},\succsim_{-w_j})=f(\succsim''_{w_j,m_j},\succsim_{-m_j,w_j})$. A contradiction then arises since $(\succsim''_{w_j,m_j},\succsim_{-m_j,w_j})$ is covered by Lemma \ref{lemma: Kew Gardens}. And we can conclude that  $(m_1,w_l),(w_1,m_l)\in f(\succsim)$.

\medskip

To conclude the proof, note that the group strategy-proofness of $f$ together with  $(m_1,w_l),(w_1,m_l)\in f(\succsim)$ imply that the royals keep their most preferred partners if they drop each other to any place in their rankings. 

\end{proof}

\begin{lemma}\label{lemma: royals must get diff commoners they want}
If $\ptop(\succsim_{m_1})=w_l$, $\ptop(\succsim_{w_1})=m_k$, and   $k\neq 1\neq l $, then $(m_1,w_l),(w_1,m_k)\in f(\succsim)$.
\end{lemma}

\begin{proof}
 Define $\succsim'_{w_1}: m_k, m_l$, $\succsim'_{m_1}: w_l,w_k$. Keeping all else equal. 
 By the Lemma \ref{lemma: the royals may choose the same partners} and group strategy-proofness the royals either get their top choices or they marry $m_k$ and $w_k$ or  $m_l$ and $w_l$. So suppose that $(m_1,w_k),(w_1,m_k)\in f(\succsim'_{m_1,w_1},\succsim_{-m_1,w_1})$. Define $\succsim'_{-m_1,w_1}$ for the commoners so that $\ptop(\succsim'_{m_k})=w_1$, $\succsim'_{w_k}=\sigma(\succsim'_{m_k})$, $\pbottom(\succsim'_{m_i})=w_1$ and $\pbottom(\succsim'_{w_i})=m_1$ for all $i\neq k$ keeping all else equal. By group strategy-proofness we have $f(\succsim'_{m_1,w_1},\succsim_{-m_1,w_1})=f(\succsim')$. A contradiction arises since $\succsim'_{-m_1,w_1}$ is covered by Lemma \ref{lemma: Kew Gardens} so that the royals must get their top choices in $f(\succsim')$. Mutatis mutandis the same arguments rule out the case that  $(m_1,w_l),(w_1,m_l)\in f(\succsim'_{m_1,w_1},\succsim_{-m_1,w_1})$ and $f(\succsim'_{m_1,w_1},\succsim_{-m_1,w_1})$ must match the royals with their top choices. Dropping the less liked partners $m_l$ and $w_k$ in their rankings, group strategy-proofness yields $f(\succsim'_{m_1,w_1},\succsim_{-m_1,w_1})=f(\succsim)$. 
\end{proof}

\begin{lemma}\label{lemma: commoner pref wins always if once}
For any fixed $\succsim_{-m_1,w_1}$ the royals are either
\begin{itemize}
    \item matched-by-default
    \item unmatched-by-default
    \item choosing sequentially dictatorially with $m_1$ going first
     \item choosing sequentially dictatorially with $w_1$ going first
\end{itemize}

\end{lemma}

\begin{proof}

Since all four regimens find the same choices when the royals either top rank each other or when they both top rank commoners, we focus on the case where exactly one royal wants to marry the other. In particular we assume $\ptop(\succsim_{w_1})=m_k$ and $\succsim_{m_1}: w_1, w_l$ for $k\neq 1 \neq l$ and note that the arguments apply by gender-neutrality to the case where $w_1$ wants to marry $m_1$. Suppose that $(m_1,w_1)\notin f(\succsim)$. 
By Lemmas \ref{lemma: the royals may choose the same partners} and \ref{lemma: royals must get diff commoners they want} and group strategy-proofness 
$f(\succsim)(m_1)\succsim_{m_1}w_l$. Since $m_1$ is by assumption not matched with the only partner he prefers to $w_l$ we have $f(\succsim)(m_1)=w_l$ and $f(\succsim)(w_1)=m_k$

Fix any $\succsim'_{m_1,w_1}$ with $\succsim'_{w_1}=m_{k'},m_{k}$ and $\succsim'_{m_1}: w_1, w_{l'}$ for $l'\neq 1 \neq k'$. 
By group strategy-proofness we may in the preceding paragraph assume w.l.o.g that $\succsim_{w_1}: m_{k},m_{k'}$ and $\succsim_{m_1}:w_1,w_l,w_{l'}$.
By  strategy-proofness  $f(\succsim'_{m_1},\succsim_{-m_1})(m_1)\in \{w_{l'},w_l\}$. By the above arguments, $(m_1,w_{l'}),(w_1,m_{k})\in f(\succsim'_{m_1},\succsim_{-m_1})$. Now swap $m_k$ and $m_{k'}$ in $w_1$'s ranking. By strategy-proofness $f(\succsim'_{m_1,w_1},\succsim_{-m_1,w_1})(w_1)\in \{m_k,m_k'\}$. By Lemmas \ref{lemma: the royals may choose the same partners} and \ref{lemma: royals must get diff commoners they want} and group strategy-proofness $(m_1,w_{l'}),(w_1,m_{k'})\in f(\succsim'_{m_1,w_1},\succsim_{-m_1,w_1})$. Dropping $m_k$ in woman $w_1$'s ranking we see that woman $w_1$ always gets her will (given $\succsim_{-m_1,w_1}$) to marry a commoner if she once gets her will. By symmetry, she would always get her will to marry $m_1$ if she once gets her will. So the only possible regimes at and $\succsim_{-m_1,w_1}$ are the two serial dictatorships as well as matched and unmatched-by-default. 

\end{proof}

\begin{lemma}
The royals are either matched-by-default or unmatched-by-default.
\end{lemma}

\begin{proof}
Fix an arbitrary $\succsim_{-m_1,w_1}$. 
By Lemma \ref{lemma: commoner pref wins always if once} the royals are either matched or unmatched-by-default or engaged in a serial dictatorship with one of the royals choosing first, the other second. 
Suppose that different profiles for the commoners were governed by different regimes. Fix two such profiles where the regimes change with the preference of one agent. Say w.l.o.g that this agent is $m_2$. So the picking regime for $w_1$ and $m_1$ is governed by different rules at $\succsim_{-m_1,w_1}$ and at $(\succsim'_{w_2},\succsim_{-m_1,w_1,m_2})$. 

The outcomes of all regimes are identical when the two royals either top rank each other or when they both rank commoners. So we fix a profile $\succsim_{m_1,w_1}$ where exactly one royal top ranks a commoner  and say that at $\succsim_{-m_1,w_1}$ the royal who wants to marry the other royal wins, while
 at $(\succsim'_{w_2},\succsim_{-m_1,w_1,m_2})$ the royals marry commoners. For concreteness assume that $\succsim_{m_1}: w_1, w_l$ and $\succsim_{w_1}: m_k$, so that the regime change from  $\succsim_{-m_1,w_1}$ to  $(\succsim'_{w_2},\succsim_{-m_1,w_1,m_2})$ is one from matched-by-default or serial dictatorship with $m_1$ going first to unmatched-by-default or serial dictatorship with $w_1$ going first. Mutatis mutandis the same arguments apply when $\succsim_{m_1,w_1}$ is such that $w_1$ wants to marry $m_1$ who in turn wants to marry a commoner. In sum  we have that $(m_1,w_1)\in f(\succsim)$ and $(w_1,m_k),(m_1,w_l)\in f(\succsim'_{m_2},\succsim_{-m_2})$.
 
 Since the regime stays fixed as long as the commoners' preferences stay fixed we may w.l.o.g assume that $m_k=m_2$ so that $f(\succsim'_{w_2},\succsim_{-m_1,w_1,m_2})(m_2)=w_1$. 
 By group strategy-proofness $f(\succsim)=f(\succsim^*_{m_2},\succsim_{-m_2})$ for any $\succsim^*_{m_2}$ that top ranks $ f(\succsim)(m_2)$ and 
 $f(\succsim''_{m_2},\succsim_{-m_1,w_1,m_2})=f(\succsim'_{w_2},\succsim_{-m_1,w_1,m_2})$ for any $\succsim''_{m_2}: w_1,  f(\succsim)(m_2)$. Now change the royal couples preferences to $\succsim'_{m_1}:w_1, f(\succsim)(m_2)$ and $\succsim'_{w_1}: m_3$. Give that the regimes stay fixed at $(\succsim^*_{m_2},\succsim_{-m_1,w_1,m_2})$ and $(\succsim''_{m_2},\succsim_{-m_1,w_1,m_2})$ we get
 \begin{eqnarray*}
 f(\succsim'_{m_1,w_1},\succsim^*_{m_2},\succsim_{-m_1,w_1,m_2})=f(\succsim)\\
 (m_1,f(\succsim)(m_2)), (w_1, m_3) \in f(\succsim'_{m_1,w_1},\succsim''_{m_2},\succsim_{-m_1,w_1,m_2}).
 \end{eqnarray*}
 
 A contradiction to strategy-proofness results since $f(\succsim'_{m_1,w_1},\succsim^*_{m_2},\succsim_{-m_1,w_1,m_2})$ matches $m_2$ with $f(\succsim)(m_2)$ which is according to $\succsim''_{m_2}$, $m_2$'s second favorite wife. Conversely, since 
$$ (m_1,f(\succsim)(m_2)), (w_1, m_3) \in f(\succsim'_{m_1,w_1},\succsim''_{m_2},\succsim_{-m_1,w_1,m_2})$$ the latter neither matches $m_2$ with his $\succsim''_{m_2}$-favorite wife $w_1$ nor with his second favorite $f(\succsim)(w_2)$. 

We in sum get that the regime with which the royals $m_1,w_1$ choose partners stays fixed for all $\succsim_{-m_1,w_1}$. By gender-neutrality the royals must use a symmetric regime when $\succsim_{-m_1,w_1}$ is gender-neutral. So neither of the two serial dictatorships can govern the choices by the royals and we must have that the regime is either matched-by-default or unmatched-by-default. 
\end{proof}

We have shown that there is a single royal couple who are either matched-by default, or unmatched by default. The notion of gender-neutrality implies that, given the matches of the royal couple, the remaining agents are engaged in a continuation mechanism which is required to be gender-neutral with respect to some symmetry of order two. Applying the same arguments to these submechanisms gives a second royal couple. Continuing in this way we get a sequence of royal couples until just four agents remain at which point our arguments break down, and any one of the four-agent mechanisms described in section \ref{appendix: four agents} can be used.

\end{document}